\documentclass[a4paper,12pt,notitlepage]{article}
\usepackage[utf8]{inputenc}
\usepackage{exscale}
\usepackage{amsmath}
\usepackage{amsthm}
\usepackage{amssymb}
\usepackage{graphicx}
\usepackage{xcolor}
\usepackage{enumerate}
\usepackage{xspace}
\usepackage{url}
\usepackage{color}
\usepackage{subfigure}
\usepackage[round]{natbib}
\usepackage{hyperref}
\hypersetup{  colorlinks=true,%
              citecolor=black,%
              filecolor=black,%
              linkcolor=black,%
              urlcolor=black,%
}
\theoremstyle{plain}
\newtheorem{theorem}{Theorem}[section]
\newtheorem*{theorem*}{Theorem}
\newtheorem{lemma}[theorem]{Lemma}
\newtheorem*{lemma*}{Lemma}
\newtheorem{proposition}[theorem]{Proposition}
\newtheorem*{proposition*}{Proposition}
\newtheorem{corollary}[theorem]{Corollary}
\newtheorem*{corollary*}{Corollary}

\newtheorem*{condition*}{Condition}
\newtheorem{definition}[theorem]{Definition}
\newtheorem*{definition*}{Definition}

\theoremstyle{remark}
\newtheorem{remark}[theorem]{Remark}
\newtheorem*{remark*}{Remark}
\newtheorem{example}[theorem]{Example}
\newtheorem*{example*}{Example}

\newcommand{\mybeginproofof}[1]{\begin{proof}[Proof of {#1}]}
\newcommand{\mybeginproof}{\begin{proof}}
\newcommand{\myendproof}{\end{proof}}

\newcommand{\mybibstyle}{\bibliographystyle{plainnat}}

\newlength{\mygraphicwidth}
\newlength{\mysubgraphicwidth}
\newcommand{\abs}[1]{|{#1}|}
\newcommand{\absfl}[1]{\left|{#1}\right|}

\newcommand{\alphaX}{\alpha_X}
\newcommand{\alphaY}{\alpha_Y}

\newcommand{\Anorm}{A_{\norm{\cdot}}}

\newcommand{\argmax}{\mathop{\mathrm{arg\,max}}\displaylimits}

\newcommand{\Axione}{A_{\xi,1}}

\newcommand{\Borel}{\mathcal{B}}

\newcommand{\Cii}{C_{i,i}}
\newcommand{\Cij}{C_{i,j}}

\newcommand{\compd}{^{(d)}}

\newcommand{\compi}{^{(i)}}

\newcommand{\compj}{^{(j)}}
\newcommand{\compone}{^{(1)}}

\newcommand{\Cov}{\mathrm{Cov}}
\newcommand{\cubr}[1]{\{{#1}\}}
\newcommand{\cubrfl}[1]{\left\{{#1}\right\}}

\newcommand{\Dii}{D_{i,i}}
\newcommand{\Dirac}[1]{\delta_{{#1}}}

\newcommand{\disteq}{\mathrel{\stackrel{\mathrm{d}}{=}}}

\newcommand{\dm}{\mathrm{d}}

\newcommand{\E}{\mathrm{E}}
\newcommand{\Ecal}{\mathcal{E}}

\newcommand{\ei}{e_i}
\newcommand{\ej}{e_j}
\newcommand{\Ellipt}{\Ecal}
\newcommand{\enquote}[1]{\lq\lq{}{#1}\rq\rq{}}
     \newlength{\laenge}
\newcommand{\eps}{\varepsilon}

\newcommand{\FR}{F_R}

\newcommand{\fu}{f_u}
\newcommand{\fuplusv}{f_{u+v}}
\newcommand{\fxialpha}{f_{\xi,\alpha}}

\newcommand{\Galphai}{G_{\alpha,i}}
\newcommand{\Galphaone}{G_{\alpha,1}}
\newcommand{\Galphatwo}{G_{\alpha,2}}
\newcommand{\Gast}{G^\ast}

\newcommand{\Gcal}{\mathcal{G}}

\newcommand{\gammaed}{\gamma_{e_d}}
\newcommand{\gammaei}{\gamma_{\ei}}
\newcommand{\gammaej}{\gamma_{\ej}}
\newcommand{\gammaeone}{\gamma_{e_1}}

\newcommand{\gammaminusei}{\gamma_{-\ei}}
\newcommand{\gammaminusej}{\gamma_{-\ej}}
\newcommand{\gammaminuseone}{\gamma_{-e_1}}
\newcommand{\gammaxi}{\gamma_\xi}

\newcommand{\gammaxione}{\gamma_{\xione}}

\newcommand{\gammaxitwo}{\gamma_{\xitwo}}
\newcommand{\Gcalalpha}{\Gcal_{\alpha}}

\newcommand{\ginv}{^{\leftarrow}}

\newcommand{\gxialpha}{g_{\xi,\alpha}}

\newcommand{\hu}{h_u}

\newcommand{\impl}{\Rightarrow}

\newcommand{\inv}{^{-1}}
\newcommand{\Law}{\mathcal{L}}

\newcommand{\minbin}{\wedge}

\newcommand{\mylefteqn}{\hspace{2em}&\hspace{-2em}}

\newcommand{\negpart}{_{-}}
\newcommand{\norm}[1]{\|{#1}\|}
\newcommand{\normfl}[1]{\left\|{#1}\right\|}

\newcommand{\nualphai}{\nu_{\alpha,i}}
\newcommand{\nuast}{\nu^{\ast}}

\newcommand{\oneby}[1]{\frac{1}{#1}}

\newcommand{\order}{\mathrel{\preceq}}
\newcommand{\orderapl}{\mathrel{\preceq_{\mathrm{apl}}}}
\newcommand{\ordercx}{\mathrel{\preceq_{\mathrm{cx}}}}
\newcommand{\orderdcx}{\mathrel{\preceq_{\mathrm{dcx}}}}
\newcommand{\orderdecx}{\mathrel{\preceq_{\mathrm{decx}}}}
\newcommand{\orderGcalalpha}{\mathrel{\preceq_{\Gcalalpha}}}
\newcommand{\ordericx}{\mathrel{\preceq_{\mathrm{icx}}}}
\newcommand{\orderlcx}{\mathrel{\preceq_{\mathrm{lcx}}}}
\newcommand{\orderplcx}{\mathrel{\preceq_{\mathrm{plcx}}}}
\newcommand{\orderpsd}{\mathrel{\preceq_{\mathrm{psd}}}}

\newcommand{\ordersm}{\mathrel{\preceq_{\mathrm{sm}}}}
\newcommand{\orderst}{\mathrel{\preceq_{\mathrm{st}}}}
\newcommand{\ordersymmcx}{\mathrel{\preceq_{\mathrm{symmcx}}}}

\newcommand{\Phiast}{\Phi^\ast}

\newcommand{\pospart}{_{+}}
\newcommand{\powalpha}{^{\alpha}}

\newcommand{\powminusalpha}{^{-\alpha}}
\newcommand{\powminusonebyalpha}{^{-1/\alpha}}
\newcommand{\powonebyalpha}{^{1/\alpha}}

\newcommand{\Prob}{\mathrm{P}}

\newcommand{\Psiast}{\Psi^\ast}
\newcommand{\Psiastone}{\Psiast_1}
\newcommand{\Psiasttwo}{\Psiast_2}
\newcommand{\PsiastX}{\Psiast_X}
\newcommand{\PsiastY}{\Psiast_Y}

\newcommand{\PsiX}{\Psi_X}

\newcommand{\R}{\mathbb{R}} 

\newcommand{\Rd}{\R^{d}}
\newcommand{\rhoalpha}{\rho_\alpha}

\newcommand{\robr}[1]{({#1})}

\newcommand{\robrfl}[1]{\left({#1}\right)}

\newcommand{\Rplus}{\R_{+}}
\newcommand{\Rplusd}{\Rplus^{d}}

\newcommand{\Sbb}{\mathbb{S}}

\newcommand{\Sd}{\Sbb^d}
\newcommand{\Sdone}{\Sbb^d_1}
\newcommand{\Sdtwo}{\Sbb^d_2}

\newcommand{\setcomp}{^\mathrm{c}}
\newcommand{\setcut}{\cap}

\newcommand{\Simp}{\Sigma}
\newcommand{\Simpd}{\Simp^d}

\newcommand{\sqbr}[1]{[{#1}]}
\newcommand{\sqbrfl}[1]{\left[{#1}\right]}

\newcommand{\sumioned}{\sum_{i=1}^{d}}

\renewcommand{\theta}{\vartheta}

\newcommand{\toinf}{\to\infty}

\newcommand{\tr}{^{\top}}

\newcommand{\unif}{\mathrm{unif}}

\newcommand{\uzer}{u_0}
\newcommand{\vagueconv}{\stackrel{\mathrm{v}}\rightarrow}
\newcommand{\Var}{\mathrm{Var}}

\newcommand{\weakconv}{\stackrel{\mathrm{w}}{\rightarrow}}

\newcommand{\xii}{\xi_i}

\newcommand{\xione}{\xi_1}

\newcommand{\xitwo}{\xi_2}

\newcommand{\BOX}{\ensuremath\Box}

\newlength{\myskipbeforeprooflength}
\newlength{\myskipafterprooflength}
\newcommand{\myskipafterproof}{\vskip\myskipafterprooflength}

\newenvironment{myproof}
{{\noindent\textit{Proof.}}}%
{\origqed\myskipafterproof\gdef\origqed{\hspace*{.1pt}\hspace*{\fill}\BOX}\par}

\newenvironment{myproofx}[1]%
{{\noindent\textit{Proof {#1}.}}}%
{\origqed\myskipafterproof\gdef\origqed{\hspace*{.1pt}\hspace*{\fill}\BOX}\par}

\makeatletter
\def\@tagformdelimstart{(}%
\def\@tagformdelimend{)}%
\def\@tagformdel{%
   \gdef\@tagformdelimstart{}%
   \gdef\@tagformdelimend{}%
}
\def\@tagformset{%
   \gdef\@tagformdelimstart{(}%
   \gdef\@tagformdelimend{)}%
}
\def\tagform@#1{%
    \maketag@@@{\@tagformdelimstart\ignorespaces#1\unskip%
    \@@italiccorr\@tagformdelimend}\@tagformset}
\def\origqed{\hspace*{.1pt}\hspace*{\fill}\BOX}
\def\qed{\ifmmode%
  \@tagformdel%
  \tag{\BOX}%
  \else%
  \hspace*{.1pt}\hspace*{\fill}\BOX%
  \fi%
  \gdef\origqed{}}
\makeatother

\begin{document}
\mybibstyle
\title{Ordering of multivariate probability distributions with respect to extreme portfolio losses}
%
%
%
\author{
Georg Mainik 
\footnote{RiskLab, Department of Mathematics, ETH Zurich} 
\and 
Ludger Rüschendorf
\footnote{Department of Mathematical Stochastics, University of Freiburg} 
}
\date{October 7, 2010}
\maketitle
\begin{abstract}
%
%
A new notion of stochastic ordering 
is introduced 
to compare 
multivariate stochastic risk models with respect 
to extreme portfolio losses. 
In the framework of multivariate regular variation comparison criteria 
are derived in terms of ordering conditions on the spectral measures, 
which allows for analytical or numerical verification in practical 
applications. 
%
Additional comparison criteria in terms of further stochastic orderings 
are derived. 
The application examples include 
worst case and best case scenarios,  
elliptically contoured distributions, 
and multivariate regularly varying models with Gumbel, Archimedean, 
and Galambos copulas. 
%
\end{abstract}

\section{Introduction}\label{sec:1}
This paper is dedicated to the comparison 
of multivariate probability 
distributions with respect to extreme portfolio losses. 
A new notion of stochastic ordering  
named \emph{asymptotic portfolio loss order} ($\orderapl$) is introduced.   
%
%
%
Specially designed for the ordering of stochastic risk models 
with respect to extreme portfolio losses, 
this notion allows to compare the inherent extreme portfolio risks 
associated with different model parameters such as correlations, 
other kinds of dependence coefficients, or diffusion parameters.  
%
%
\par
In a recent paper of \cite{Mainik/Rueschendorf:2010} the notion of 
\emph{extreme risk index} has been introduced in the framework of 
multivariate regular variation. This index, denoted by $\gammaxi$, 
is a functional of the vector $\xi$ of portfolio weights and of the 
characteristics of the multivariate regular variation of $X$ given 
by the tail index $\alpha$ and the spectral measure $\Psi$. 
It measures the sensitivity of the portfolio loss 
to extremal 
events and characterizes the probability distribution of extreme losses. 
In particular, it serves to determine the optimal portfolio diversification 
with respect to extreme losses. 
Within the framework of multivariate regular variation 
the notion of asymptotic portfolio loss ordering introduced in this paper 
is tightly related to model comparison in terms of 
the extreme risk index $\gammaxi$. 
%
Thus this paper can be seen as a supplement of the previous one, 
allowing to order multivariate risk models with respect to    
their extremal portfolio loss behaviour. 
\par
In Section \ref{sec:2} of the present paper we introduce the 
asymptotic portfolio loss order $\orderapl$ and highlight some relationships 
to further well-known ordering notions. 
It turns out that even strong dependence and convexity 
orders do not imply the asymptotic portfolio loss order in general. 
We present counter-examples, based on the the inversion of diversification 
effects in models with infinite loss expectations. 
Another example of particular interest discussed here is given 
by the elliptical distributions. 
%
In this model family we establish a precise criterion for the asymptotic 
portfolio loss order, which perfectly accords with the classical 
results upon other well-known order relations. 
%
Section \ref{sec:3} is devoted to multivariate regularly varying models. 
We discuss the relationship between the asymptotic portfolio loss order 
and the comparison of the extreme risk index and characterize $\orderapl$ 
in terms of a suitable ordering of the canonical spectral measures. 
These findings allow to establish sufficient conditions for 
$\orderapl$ in terms of spectral measures, which can be verified 
by analytical or numerical methods.  
In particular, we characterize the dependence structures that yield 
the best and the worst possible diversification effects for a multivariate 
regularly varying risk vector $X$ in $\Rplusd$ with tail index $\alpha$.
For $\alpha\ge1$ the best case is given by the asymptotic independence and  
the worst case is the asymptotic comonotonicity. 
The result for $\alpha\le 1$ is exactly the opposite 
(cf.\ Theorem~\ref{thm:3.8} and Corollary~\ref{cor:3.10}).
Restricting $X$ to $\Rplusd$ means that $X$ represents only the losses, whereas 
the gains are modelled separately. 
This modelling approach is particularly suitable for applications in 
insurance, operational risk, and credit risk.
If $X$ represents both losses and gains, these results remain valid if the 
extremal behaviour of the gains is weaker than that of the losses, so 
that there is no loss-gain compensation for extremal events.
In Section \ref{sec:4} we discuss the interconnections between $\orderapl$ 
or ordered canonical spectral measures and other well-known 
notions of stochastic ordering. 
Ordering of canonical spectral measures allows to conclude $\orderapl$
from the (directionally) 
convex or the supermodular order. 
It is not obvious how to obtain this implication in a general setting.
Finally, in Section~\ref{sec:5} we present  a series of examples 
with graphics illustrating the numerical results upon the ordering 
of spectral measures.
%
The relationship to spectral measures provides a useful numerical tool to 
establish $\orderapl$ in practical applications.
%
%
\section{Asymptotic portfolio loss ordering}\label{sec:2}
To compare stochastic risk models with respect to extreme portfolio losses, 
we introduce the asymptotic portfolio loss order $\orderapl$. 
This order relation is designed 
for the analysis of the asymptotic 
diversification effects and the identification of models that generate 
portfolio risks with stronger extremal behaviour. 
\par
Before stating the definition, some basic notation is needed. 
Focusing on risks, let $X$ be a \emph{random loss vector} with values in 
$\Rd$, i.e., let positive values of the components $X\compi$, $i=1,\ldots,d$,
represent losses and let negative values of $X\compi$ represent gains 
of some risky assets. 
Following the intuition of diversifying a unit capital over several assets,
we restrict the set of portfolios to the unit simplex in $\Rd$:  
\[
\Simpd := \cubrfl{\xi\in\Rplusd: \sumioned\xii=1 }
\ldotp
\]
The portfolio loss resulting from a random vector $X$ and the portfolio $\xi$ 
is given by the scalar product of $\xi$ and $X$. 
In the sequel it will be denoted by  $\xi\tr X$.
\begin{definition}
Let $X$ and $Y$ be $d$-dimensional random vectors. 
Then $X$ is called smaller than $Y$ in 
\emph{asymptotic portfolio loss order}, $X\orderapl Y$, 
if 
\begin{equation}\label{eq:2.1}
\forall\xi\in\Simpd
\quad
\limsup_{t\to\infty}  
\frac{\Prob\cubr{\xi\tr X> t}}{\Prob\cubr{\xi\tr Y\ge t}} \le 1
\ldotp
\end{equation}
Here, $\frac00$ is defined to be 1.
\end{definition}
\begin{remark}\label{rem:2.1}
\begin{enumerate}[(a)]
\item
Although designed for random vectors, $\orderapl$ is also defined for 
random variables. In this case, the portfolio set has only one element, 
$\Simp^1=\cubr{1}$.   
\item\label{item:rem:2.1.b}
It is obvious that  $\orderapl$ is invariant under componentwise rescaling. 
Let $vx$ denote the componentwise product of $v,x\in \Rd$: 
\begin{equation}\label{eq:apl.2}
vx:= (v\compi x\compi,\dots,v\compd x\compd),
\end{equation}
Then it is easy to see that $ X\orderapl Y$ implies $vX \orderapl vY$ for 
all $v\in\Rplusd$.
Hence condition~\eqref{eq:2.1} can be equivalently stated for $\xi\in\Rplusd$. 
\end{enumerate}
\end{remark}
\par
The ordering statement $X\orderapl Y$ means that for all portfolios 
$\xi\in\Simpd$ 
the portfolio loss $\xi\tr X$ is asymptotically smaller $\xi\tr Y$. 
Thus $\orderapl$ concerns only the extreme portfolio losses. 
In consequence, this order relation is weaker than the (usual) 
stochastic ordering $\orderst$ of the portfolio losses:
\begin{equation}\label{eq:2.2}
\xi\tr X \orderst \xi\tr Y \text{ for all } \xi\in\Simpd \text{ implies } X\orderapl Y.
\end{equation}
Here, for real random variables $U$, $V$ the \emph{stochastic ordering} 
$U \orderst V$ is defined by 
\begin{equation}\label{eq:2.3a}
\forall t\in\R\quad \Prob\cubr{U>t}\le\Prob\cubr{V>t}.
\end{equation}
\par
Some related, well-known stochastic orderings  
\citep[cf.][]{Mueller/Stoyan:2002,Shaked/Shanthikumar:1997} 
are collected in the following list. Remind that $f:\Rd\to \R$ 
is called \emph{supermodular} if 
\begin{equation}\label{eq:2.3b}
\forall x,y\in\Rd
\quad
f(x\wedge y)+f(x\vee y)\ge f(x)+f(y)
\ldotp
\end{equation}
\begin{definition}\label{def:2.2}
Let $X$, $Y$ be random vectors in $\Rd$. Then $X$ is said to be smaller than $Y$ in
\begin{enumerate}[(a)]
\item \emph{(increasing) convex order}, 
$X \ordercx Y$ ($X\ordericx Y$), if $\E f(X) \le \E f(Y)$ for all (increasing) convex functions $f:\Rd\mapsto \R$ such that the expectations exist; 
\item
\emph{linear convex order}, $X \orderlcx Y$, if 
$\xi\tr X \ordercx \xi\tr Y$ for all $\xi\in\Rd$;
\item
\emph{positive linear convex order}, $X \orderplcx Y$, 
if $\xi\tr X \ordercx \xi\tr Y$ for all $\xi\in\Rplusd$;
\item \emph{supermodular order} $X\ordersm Y$,  if 
$\E f(X) \le \E f(Y)$ for all supermodular functions $f:\Rd\to\R$ such
that the expectations exist;
\item
\emph{directionally convex order}, $X\orderdcx Y$, if 
$\E f(X) \le \E f(Y)$ for all directionally  convex, i.e., supermodular and componentwise convex functions  
$f:\Rd\to\R$ such that the expectations exist.
\end{enumerate}
\end{definition}
\par
The stochastic orderings listed in Definition \ref{def:2.2} are useful 
for describing the risk induced by larger diffusion (convex risk) as well as 
the risk induced by positive dependence 
(supermodular and directionally convex). 
The following implications are known to hold generally for random 
vectors $X$, $Y$ in $\Rd$: 
\begin{enumerate}[(a)]
\item $(X\ordersm Y)_{\phantom{icx}\kern-2ex} \impl %
(X\orderdcx Y)_{\phantom{l}\kern-.5ex} \impl (X\orderplcx Y)$
\item $(X\ordercx Y)_{\phantom{ism}\kern-2ex} \impl %
(X\orderlcx Y)_{\phantom{d}\kern-.5ex} \impl (X\orderplcx Y)$
\item $(X\ordericx Y)_{\phantom{sm}\kern-2ex} \impl %
(X\orderplcx Y)$
\end{enumerate}
\begin{remark}\label{rem:apl.1}
\begin{enumerate}[(a)]
\item\label{item:apl.1}
It is easy to see that the usual stochastic order $\orderst$ implies 
$\orderapl$ in the univariate case.
\item\label{item:apl.2}
In spite of being strong risk comparison orders, the order relations 
outlined in Definition~\ref{def:2.2} do not imply $\orderapl$ in general. 
For instance, it is known that the comonotonic dependence structure is 
the worst case with respect to the strong supermodular ordering $\ordersm$,
whereas it is not necessarily the worst case with respect to $\orderapl$
(cf.\ Examples~\ref{ex:6} and \ref{ex:2}). 
\end{enumerate}
\end{remark}
\par
The following proposition helps to establish sufficient criteria 
for $\orderapl$ in the univariate case. 
To obtain multivariate results, 
it can be separately applied to each portfolio loss $\xi\tr X$ for 
$\xi\in\Simpd$.   
\par
\begin{proposition}\label{prop:2.3}
Let $R_1$, $R_2\ge 0$ be real random variables and let $V$ be a real random variable independent of $R_i$, $i=1,2$. 
\begin{enumerate}[(a)]
\item
\label{item:prop2.3b} 
If $R_1\orderapl R_2$ and $V < K$ 
for some constant $K$, then 
\begin{equation}\label{eq:2.7}
R_1V\orderapl R_2V
\ldotp
\end{equation}
\item  \label{item:prop2.3a} 
If $R_1\orderst R_2$, then
%
\begin{equation}
\label{eq:2.5}
\robr{R_1V}\pospart  \orderst \robr{R_2V}\pospart 
\quad 
\text{and}
\quad
\robr{R_2V}\negpart \orderst \robr{R_1V}\negpart
\ldotp
\end{equation}
In addition, if $V$ and $R_i$ are integrable and $EV \ge 0$, then 
\begin{equation}\label{eq:2.6}
R_1V \ordericx R_2V
\ldotp
\end{equation}
Moreover, if $EV=0$, then $R_1V \ordercx R_2V$.
%
\end{enumerate}
\end{proposition}
\par
\begin{myproof}
%
\par Part~(\ref{item:prop2.3b}). 
Since $R_1V \orderapl R_2V$ is trivial for $V \le 0$, we 
assume that $\Prob\cubr{V>0}>0$. 
Hence $V\le K$ implies for all $t>0$
\begin{align} 
\Prob\cubrfl{R_1V>t}  
&= \nonumber
\int_{(0,K)} \Prob\cubrfl{R_1>{t}/{v}} \dm\Prob^V(v) \\
&=\label{eq:2.10b}
\int_{(0,K)} f\robrfl{{t}/{v}} \Prob\cubrfl{R_2>{t}/{v}} \dm\Prob^V(v), 
\end{align}
%
where 
\[
f(z):=\frac{\Prob\cubr{R_1>z}}{\Prob\cubr{R_2>z}}
\ldotp
\]
An obvious consequence of \eqref{eq:2.10b} is the inequality 
\begin{equation}
\Prob\cubr{R_1V > t} 
\le\label{eq:2.10c}
\sup \cubrfl{ f(z): z > {t}/{K}} \cdot \Prob\cubrfl{R_2V > t}
\end{equation}
Since $R_1\orderapl R_2$ is equivalent to $\limsup_{z\toinf} f(z)\le 1$,
we obtain 
\[
\limsup_{t\toinf}\frac{\Prob\cubr{R_1V>t}}{\Prob\cubr{R_2V>t}}\le 1
\ldotp
\] 
\par
Part (\ref{item:prop2.3a}). 
%
By the well-known coupling principle for the stochastic ordering $\orderst$
we may assume without loss of generality that $R_1 \le R_2$ 
pointwise on the underlying probability space.
This implies 
\[
\Prob\cubr{R_1 V > t} \le \Prob\cubr{R_2 V > t}
,\quad t \ge 0,
\]
and, similarly,
\[
\Prob\cubr{R_1 V \le t} \le \Prob\cubr{R_2 V \le t}
,\quad t \le 0
\ldotp
\]
In consequence we obtain~\eqref{eq:2.5}.
%
\par
From the proof of \eqref{eq:2.5} 
it follows that the distribution functions 
of the products $R_iV$, $i=1,2$, satisfy the cut criterion of Karlin--Novikov 
(cf.\ \citealp{Shaked/Shanthikumar:1994}, Theorem 2.A.17 and 
\citealp{Mueller/Stoyan:2002}, Theorem 1.5.17) 
Hence we obtain
\begin{equation}\label{eq:2.9}
R_1V \ordericx R_2V
\ldotp
\end{equation}
If $EV=0$, then $E\sqbr{R_1V}=E\sqbr{R_2V}$ and therefore 
\begin{equation}\label{eq:2.10a}
R_1V\ordercx R_2V
\ldotp
\end{equation}
\end{myproof}
\begin{remark}\label{rem:2.3}
\begin{enumerate}[(a)]
\item %
Note that 
\eqref{eq:2.5}
implies (without assuming the existence of moments) 
that $\robr{R_2V}\pospart \orderdecx \robr{R_1V}\pospart$ 
where $\orderdecx$ denotes the \emph{decreasing convex order}. 
Similarly one obtains 
$\robr{R_2V}\negpart \ordericx \robr{R_1V}\negpart$
\item 
If $f(t):={\Prob\cubr{R_1>t}}/{\Prob\cubr{R_2>t}} \le C < \infty$ and 
$R_1\orderapl R_2$, then $R_1V\orderapl R_2V$. 
\item 
A related problem is the ordering of products $RV_i$ for $R\ge 0$ with  
$V_1$ and $V_2$ independent of $R$.
In the special case when $R$ is \emph{regularly varying} with \emph{tail index} 
$\alpha>0$, i.e., 
\begin{equation}\label{eq:apl.5}
\lim_{t\toinf}
\frac{\Prob\cubr{R>tx}}{\Prob\cubr{R>t}} 
=x\powminusalpha
,\quad 
x>0,
\end{equation}
exact criteria for $\orderapl$ can be obtained from Breiman's Theorem 
\citep[cf.][Proposition 7.5]{Resnick:2007}. 
If $\E \robr{V_i}\pospart^{\alpha+\eps} <\infty$ for $i=1,2$ and 
some $\eps>0$, then 
\[
\lim_{t\toinf}\frac{\Prob\cubr{RV_i>t}}{\Prob\cubr{R>t}} 
= 
E\sqbrfl{\robr{V_i}\pospart\powalpha}
\ldotp
\]
This yields 
\[
\lim_{t\toinf}
\frac{\Prob\cubr{RV_1>t}}{\Prob\cubr{RV_2>t}} 
= 
\frac
{\E\sqbrfl{\robr{V_1}\pospart\powalpha}}
{\E\sqbrfl{\robr{V_2}\pospart\powalpha}}
\ldotp
\]
%
%
\end{enumerate}
\end{remark}
An important class of stochastic models with various applications are 
\emph{elliptical distributions}, 
which are natural generalizations of multivariate 
normal distributions. 
A random vector $X\in\Rd$ is called elliptically distributed, 
if there exist $\mu\in\Rd$ and a $d\times d$ matrix $A$ such that  
$X$ has a representation of the form
\begin{equation}\label{eq:2.11}
X \disteq \mu + RAU, 
\end{equation}
where $U$ is uniformly distributed on the Euclidean unit sphere $\Sdtwo$,  
\[
\Sdtwo=\cubrfl{x\in\Rd :  \norm{x}_2 =1},
\] 
and $R$ is a non-negative random variable independent of $U$. 
By definition we have 
\begin{equation}\label{eq:2.12}
E\norm{X}_2^2 <\infty \Leftrightarrow E R^2<\infty,
\end{equation}
and in this case   
\begin{equation}\label{eq:2.13}
\Cov(X)=\Var(R) A A\tr
\ldotp 
\end{equation}
The matrix $C:= A A\tr$ is unique except for a constant factor and 
is also called the \emph{generalized covariance matrix} of $X$. 
%
%
We denote the elliptical distribution constructed 
according to~\eqref{eq:2.11} by $\Ellipt(\mu,C,\FR)$, 
where $\FR$ is the distribution of $R$.
\par
A classical stochastic ordering result going back to 
\cite{Anderson:1955} and \cite{Fefferman/Jodeit/Perlman:1972} 
\citep[cf.][p.~70]{Tong:1980} 
says that \emph{positive semidefinite ordering}
of the generalized covariance matrices $C_1 \orderpsd C_2$, defined as
\begin{equation}
\label{eq:2.13a}
\forall \xi\in\Rd \quad \xi\tr C_1 \xi \le \xi\tr C_2 \xi,
\end{equation}
implies symmetric convex ordering if  
the location parameter $\mu$ and the distribution $\FR$ of the radial 
factor are fixed:  
\begin{equation}\label{eq:2.14}
\Ellipt(\mu, C_1,\FR) \ordersymmcx \Ellipt(\mu, C_2, \FR)
\ldotp
\end{equation}
It is also known that for elliptical random vectors $X\sim \Ellipt(\mu,C,\FR)$ 
the multivariate distribution function 
$F(x):=\Prob\cubr{X_1\le x_1,\dots,X_d\le x_d}$ is increasing in $\Cij$ for $i\not=j$, where $C=(\Cij)$ \citep[see, e.g.,][Theorem 2.21]{Joe:1997}.
\par
The following result is concerned 
with the asymptotic portfolio loss ordering $\orderapl$ for elliptical 
distributions. 
%
\begin{theorem}\label{theo:2.4}
Let $X\disteq \mu_1+R_1A_1U$,  $Y\disteq\mu_2+R_2A_2U$ be elliptically distributed 
with generalized covariances $C_i:=A_iA_i\tr$. If 
\begin{equation}\label{eq:2.15}
\mu_1\le \mu_2, \enskip 
R_1\orderapl R_2,
\end{equation}
and
\begin{equation}\label{eq:2.15a}
\forall \xi\in\Simpd \quad \xi\tr C_1\xi \le \xi\tr C_2 \xi,
\end{equation}
then
\begin{equation}\label{eq:2.16}
X\orderapl Y.
\end{equation}
\end{theorem}
\par
\begin{myproof}
It suffices to show that $\xi\tr Y \orderapl \xi\tr Y$ for an 
arbitrary portfolio $\xi\in\Simpd$. 
Furthermore, without loss of generality we can assume $\mu_1=\mu_2=0$.
For $i=1,2$ and $\xi\in\Simpd$ denote 
\[
a_i = a_i(\xi):=  \robrfl{\xi \tr C_i \xi}^{1/2}
\] 
and 
\[
v_i = v_i(\xi) := \frac{\xi\tr A_i}{a_i}
\ldotp
\]
Then, by definition of elliptical distributions, we have
\begin{equation}\label{eq:2.17}
\xi\tr X \disteq R_1 a_1 v_1 U
\quad\text{and}\quad
\xi\tr Y \disteq R_2 a_2 v_2 U
\ldotp
\end{equation}
%
%
Since the vectors $v_i=v_i(\xi)$ have unit length by construction, 
the random variables $v_i U$ are orthogonal projections of $U\sim \unif(S_2^d)$ 
on vectors of unit length. 
Symmetry arguments yield that the distribution of $v_i U$ is independent of 
$v_i$ and that $v_iU \disteq (1,0,\ldots,0)\tr U=U\compone$. 
\par
Thus we have 
\[
\xi\tr X \disteq a_1 R_1 V
\quad\text{and}\quad
\xi\tr Y\disteq a_2 R_2 V
\] 
with $V:=U\compone$.
By assumption we have $a_1\le a_2$ and $R_1\orderapl R_2$. 
Applying Proposition \ref{prop:2.3}(\ref{item:prop2.3b}) 
we obtain $\xi\tr X\orderapl \xi\tr Y$. 
\end{myproof}
\par
\begin{remark}\label{rem:2.6}
\begin{enumerate}[(a)]
\item\label{item:rem:2.6.a}
It should be noted that condition~\eqref{eq:2.15a} is indeed weaker than 
\eqref{eq:2.13a}. 
Let $-1 < \rho_1 < \rho_2 <1$ and consider covariance matrices 
\[
C_i:=
\robrfl{
\begin{array}{cc}
1 & \rho_i\\
\rho_i & 1
\end{array}
}
,\quad 
i=1,2
\ldotp
\]
Straightforward calculations show that $C_i$ satisfy~\eqref{eq:2.15a}, but 
not~\eqref{eq:2.13a}.
\item
For subexponentially distributed 
$R_i$ the assumption $\mu_1\le \mu_2$ in \eqref{eq:2.15} can be omitted.
\end{enumerate}
\end{remark}
%
\section{Multivariate regular variation: $\orderapl$ in terms of spectral measures} 
\label{sec:3}
\par
This section is concerned with the characterization of the asymptotic 
portfolio loss order $\orderapl$ in the framework of multivariate regular 
variation. The results obtained here highlight the influence of the tail 
index $\alpha$ and the spectral measure $\Psi$ on $\orderapl$,  
with primary focus put on  dependence structures captured by $\Psi$. 
It is shown that $\orderapl$ corresponds to a family of order relations  
on the set of canonical spectral measures and that these order relations 
are intimately related to the extreme risk index $\gammaxi$ introduced 
in \citet{Mainik/Rueschendorf:2010} and \citet{Mainik:2010}. 
\par
The main result of this section is stated in Theorem~\ref{theo:3.4}, 
providing criteria for $X \orderapl Y$ in terms of componentwise ordering 
$X\compi \orderapl Y\compi$ for $i=1,\ldots,d$ 
and ordering of canonical spectral measures. 
A particular consequence of these criteria is the 
characterization of the dependence structures that 
yield the best and the worst possible diversification effects for 
random vectors in $\Rplusd$ 
(cf.\ Theorem~\ref{thm:3.8} and Corollary~\ref{cor:3.10}).
Another application concerns elliptical distributions. Combining  
Theorem~\ref{theo:3.4} with results on $\orderapl$ obtained in 
Theorem~\ref{theo:2.4}, we obtain ordering of the corresponding 
canonical spectral measures. 
\par
Recall the notions of regular variation. 
In the univariate case it can be defined separately for the lower 
and the upper tail of a random variable via~\eqref{eq:apl.5}. 
%
%
A random vector $X$ taking values in $\Rd$ is called 
\emph{multivariate regularly varying} with tail index $\alpha\in(0,\infty)$
if  there exist a sequence $a_n\toinf$ and a (non-zero) Radon measure $\nu$ on 
the Borel $\sigma$-field $\Borel\robr{[-\infty,\infty]^d\setminus\cubr{0}}$ 
such that $\nu\robr{[-\infty,\infty]^d \setminus \Rd}=0$ and, 
as $n\toinf$,
\begin{equation}
\label{eq:29}
\index{$\nu$}
n \Prob^{\,a_n\inv X} \vagueconv \nu
\text{ on }\Borel\robr{[-\infty,\infty]^d\setminus\cubr{0}},
\end{equation}
where $\vagueconv$ denotes the 
\emph{vague convergence} of Radon measures 
and $\Prob^{\,a_n\inv X}$ is the probability distribution of
$a_n\inv X$.
\par
It should be noted that random vectors with non-negative components 
yield limit measures $\nu$ that are concentrated on 
$[0,\infty]^d\setminus\cubr{0}$. 
Therefore multivariate regular variation in this special case can also 
be defined by vague convergence on $\Borel([0,\infty]^d\setminus\cubr{0})$.
\par
Many popular distribution models are multivariate regularly  varying. 
In particular, according to \citet{Hult/Lindskog:2002}, 
multivariate regular variation of an elliptical distribution 
$\Ellipt\robr{\mu,C,\FR}$ is equivalent to the regular variation of the 
radial factor $R$ and the tail index $\alpha$ is inherited from $R$.
Other popular examples are obtained by endowing regularly varying margins 
$X\compi$ with an appropriate copula 
\citet[cf.][]{Wuethrich:2003, Alink/Loewe/Wuethrich:2004, Barbe/Fougeres/Genest:2006}
\par
For a full account of technical details related to the notion of 
multivariate regular variation, vague convergence, and 
the Borel $\sigma$-fields on the punctured spaces 
$[-\infty,\infty]^d\setminus\cubr{0}$ and $[0,\infty]^d\setminus\cubr{0}$  
the reader is referred to \citet{Resnick:2007}. 
\par
It is well known that the limit measure $\nu$ obtained in~\eqref{eq:29}
is unique except for a constant factor, has a singularity in the origin
in the sense that 
$\nu\robr{(-\eps,\eps)^d}=\infty$ for any $\eps>0$,  
and exhibits the scaling property 
\begin{equation}
\label{eq:30}
\nu(tA)=t\powminusalpha\nu(A)
\end{equation} 
for all sets $A\in\Borel\robrfl{[-\infty,\infty]^d\setminus\cubr{0}}$ that
are bounded away from $0$. 
\par
It is also well known that~\eqref{eq:29} implies that the random variable 
$\norm{X}$ with an arbitrary norm $\norm{\cdot}$ on $\Rd$ is 
univariate regularly varying with tail index $\alpha$.
Moreover,
the sequence $a_n$ can always be chosen as 
\begin{equation}
\label{eq:181}
a_n:=F_{\norm{X}}\ginv(1-1/n),
\end{equation}
where $F_{\norm{X}}\ginv$ is the quantile function of 
$\norm{X}$. The resulting limit measure $\nu$ 
is normalized on the set $\Anorm:=\cubr{x\in\Rd: \norm{x}>1}$ by 
\begin{equation}
\label{eq:182}
\nu\robrfl{\Anorm}=1
\ldotp
\end{equation}
\par
Thus, after normalizing $\nu$ by~\eqref{eq:182}, 
the scaling relation~\eqref{eq:30} yields an equivalent rewriting of 
the multivariate regular variation condition~\eqref{eq:29} 
in terms of weak convergence:
\begin{equation}
\label{eq:34}
\Law\cubrfl{t\inv X\,|\,\norm{X}>t}
\weakconv
\nu|_{\Anorm}
\text{ on } 
\Borel\robrfl{\Anorm}
\end{equation}
for $t\toinf$,
where $\nu|_{\Anorm}$ is the restriction of $\nu$ to the set $\Anorm$. 
\par
Additionally to~\eqref{eq:29}
it is assumed that the limit measure $\nu$ is 
non-degen\-erate
in the 
following sense:
\begin{equation}
\label{eq:4}
\nu\robrfl{\cubrfl{x\in\Rd: \absfl{x\compi}> 1}} >0
,\quad i=1,\ldots,d
\ldotp
\end{equation}
This assumption ensures that
all asset losses $X\compi$ are relevant for the extremes of the portfolio loss 
$\xi\tr X$. If~\eqref{eq:4} is satisfied in the upper tail region, i.e., if 
\begin{equation}
\label{eq:4a}
\nu\robrfl{\cubrfl{x\in\Rd: x\compi> 1}} >0 
,\quad i=1,\ldots,d,
\end{equation} 
then  $\nu$ also characterizes the asymptotic distribution
of the componentwise maxima
$M_n:=\robr{M\compone,\ldots,M\compd}$ with 
$M\compi:=\max\cubr{X_1\compi,\ldots,X_n\compi}$
by the limit relation 
\begin{equation}
\label{eq:164}
\Prob\cubrfl{a_n\inv M_n\in[-\infty,x]} \weakconv
\exp\robrfl{-\nu\robrfl{[-\infty,\infty]^d\setminus[-\infty,x]}}
\end{equation}
for $x\in(0,\infty]^d$. 
Therefore $\nu$ is called 
\emph{exponent measure}. 
For more details concerning the asymptotic distributions of maxima
the reader is referred to~\citet{Resnick:1987} 
and \citet{de_Haan/Ferreira:2006}.
\par
Another consequence of the scaling property~\eqref{eq:30} is the  
product representation of $\nu$ in polar coordinates 
\[
(r,s):=\tau(x):=(\norm{x},\norm{x}\inv x)
\] 
with respect to an arbitrary norm $\norm{\cdot}$ on $\Rd$.
The induced 
measure $\nu^\tau:=\nu\circ\tau\inv$ necessarily satisfies
\begin{equation}
\label{eq:28}
\nu^\tau=c\cdot\rhoalpha\otimes\Psi
\end{equation}
with the constant factor 
\[
c=\nu\robrfl{\Anorm}
>0,
\]  
the measure $\rhoalpha$ on $(0,\infty]$ defined by 
\begin{equation}
\label{eq:176}
\rhoalpha((x,\infty]):=x\powminusalpha,
\quad 
x\in(0,\infty],
\end{equation} 
and a probability measure $\Psi$ on the unit sphere $\Sd_{\norm{\cdot}}$
with respect to $\norm{\cdot}$,
\[ 
\Sd_{\norm{\cdot}}:=\cubrfl{s\in\Rd : \norm{s} = 1}
\ldotp
\]
The measure $\Psi$ is called 
\emph{spectral measure} 
of $\nu$ or $X$.
Since the term \enquote{spectral measure} is already used in other areas, 
$\Psi$ is also referred to as 
\emph{angular measure}.
In the special case of $\Rplusd$-valued random vectors $X$ it
may be convenient to reduce the domain of $\Psi$ to 
$\Sd_{\norm{\cdot}}\setcut\Rplusd$. 
\par
Although the domain of the spectral measure $\Psi$ depends on the
norm $\norm{\cdot}$ underlying the polar coordinates, the 
representation~\eqref{eq:28} is norm-independent in the following sense:
if~\eqref{eq:28} holds for some norm $\norm{\cdot}$, then it also holds for 
any other norm $\norm{\cdot}_\diamond$ that is equivalent to $\norm{\cdot}$. 
The tail index $\alpha$ is the same and the spectral measure $\Psi_\diamond$ 
on the unit sphere $\Sd_\diamond$ corresponding to $\norm{\cdot}_\diamond$ 
is obtained from $\Psi$ by the following transformation:
\[
\Psi_\diamond=\Psi^T,\quad T(s):=\norm{s}_\diamond\inv s
\ldotp
\]
%
\par
Finally, it should be noted that multivariate regular variation of 
the loss vector $X$ is intimately related with the univariate regular variation 
of portfolio losses $\xi\tr X$.
As shown in \citet{Basrak/Mikosch/Davis:2002}, 
multivariate regular variation of $X$ 
implies existence of a portfolio vector $\xi_0\in\Rd$ such that $\xi_0 \tr X$ 
is regularly varying with tail index $\alpha$ and any 
portfolio loss $\xi\tr X$ satisfies
\begin{equation}
\label{eq:192}
\lim_{t\toinf}
\frac{\Prob\cubrfl{\xi\tr X >t}}{\Prob\cubrfl{\xi_0\tr X >t}} 
=c(\xi,\xi_0)
\in [0,\infty)
\ldotp
\end{equation}
This means that all portfolio losses $\xi\tr X$ are either regularly
varying with tail index $\alpha$ or asymptotically negligible 
compared to $\xi_0\tr X$. 
\par
Moreover, it is also worth a remark 
that for $\Rplusd$-valued random vectors $X$ 
the converse implication is true in the sense that~\eqref{eq:192} 
and univariate regular variation of $\xi_0\tr X$ 
imply multivariate regular variation of the random vector $X$.
This sort of Cram\'er-Wold theorem was established in 
\citet{Basrak/Mikosch/Davis:2002} and \citet{Boman/Lindskog:2009}.
\par
Under the assumption of multivariate regular variation of $X$ 
the \emph{extreme risk index} $\gamma_\xi = \gamma_\xi(X)$ 
is defined as 
%
\begin{equation}\label{eq:3.2}
\gamma_\xi(X)=\lim_{t\toinf} \frac{\Prob\cubr{\xi\tr X>t}}{\Prob\cubr{\norm{X}_1>t}}.
\end{equation}
In \citet{Mainik/Rueschendorf:2010} the random vector $X$ is restricted to 
$\Rplusd$ and the portfolio vector $\xi$ is restricted to $\Simpd$. 
The general case with $X$ in $\Rd$ and possible negative portfolio 
weights, i.e., short positions, is considered in \citet{Mainik:2010}. 
Normalizing the exponent measure $\nu$ by~\eqref{eq:182},
one obtains 
\begin{equation}\label{eq:3.1}
\gamma_\xi(X)=\nu\robrfl{\cubrfl{x\in\Rd: \xi\tr x > 1}}
\ldotp
\end{equation}
Rewriting this representation in terms of the spectral measure $\Psi$ 
and the tail index $\alpha$ yields
\begin{equation}\label{eq:apl.1}
\gammaxi
=
\int_{\Sdone}\robrfl{\xi\tr s}\pospart\powalpha \dm \Psi(s)
\ldotp
\end{equation}
Denoting the integrand by $\fxialpha$, we will write this representation 
as $\gammaxi=\Psi\fxialpha$.   
%
%
\par
The extreme risk index $\gamma_\xi(X)$ allows to compare the risk of different 
portfolios. It is easy to see that \eqref{eq:3.2} implies
\begin{equation}\label{eq:3.4}
\lim_{t\toinf} \frac{\Prob\cubr{\xione\tr X>t}}{\Prob\cubr{\xitwo\tr X>t}} = \frac{\gammaxione(X)}{\gammaxitwo(X)}.
\end{equation}
Thus, by construction, 
ordering of the extreme risk index $\gammaxi$ is related to the 
asymptotic portfolio loss order $\orderapl$. 
\par
However, designed for the comparison of different portfolio risks within one 
model, the extreme risk index $\gammaxi$ cannot be directly applied 
to the comparison of different models. 
The major problem is the standardization by $\Prob\cubr{\norm{X}_1>t}$ in 
\eqref{eq:3.2}. Indeed, since $\Prob\cubr{\norm{X}_1>t}$ also depends on the 
spectral measure $\PsiX$ of $X$, criteria for $\orderapl$ in terms of 
$\gammaxi$  demand the specification of the limit
\[
\lim_{t\toinf} 
\frac{\Prob\cubr{\norm{X}_1>t}}{\Prob\cubr{\norm{Y}_1>t}}
\ldotp
\]
\par 
Another technical issue arises from the invariance of $\orderapl$ under 
componentwise rescalings. Since the spectral measure $\Psi$ does not exhibit 
this property, ordering of spectral measures needs additional normalization 
of margins that makes it consistent with $\orderapl$. To solve these problems,
 we use an alternative representation of $\gammaxi$ in terms of the 
so-called canonical spectral measure $\Psiast$, 
which has standardized marginal weights.
\par
This representation is closely related to the asymptotic risk aggregation 
coefficient discussed by \cite{Barbe/Fougeres/Genest:2006}. 
Furthermore, the link between the canonical spectral measure and 
extreme value copulas 
allows to transfer ordering results for copulas into the $\orderapl$ 
setting. These results are presented in Section~\ref{sec:4}.
\par
To reduce the problem to the essentials,  
we start with the observation that $\orderapl$ is trivial for  
multivariate regularly varying random vectors with different 
tail indices and non-degenerate portfolio losses. 
\par
\begin{proposition}\label{prop:3.1}
Let $X$ and $Y$ be multivariate regularly varying on $\Rd$ and assume that $\gamma_\xi(Y)>0$ for all $\xi\in\Simpd$. 
\begin{enumerate}[(a)]\label{item:prop.3.1a}
\item If
\begin{equation}\label{eq:3.5}
\lim_{t\toinf} \frac{\Prob\cubr{\norm{X}_1>t}}{\Prob\cubr{\norm{Y}_1>t}} = 0,
\end{equation}
then $X \orderapl Y$.
\vspace{0.5em}
\item If $\alphaX>\alphaY$, then $X \orderapl Y$.
\end{enumerate}
\end{proposition}
\par
\begin{myproof}
\begin{enumerate}[(a)]
\item %
Using relation \eqref{eq:3.2} we obtain 
\begin{align*}
\mylefteqn
\limsup_{t\toinf}
\frac{\Prob\cubrfl{\xi\tr X > t}}{\Prob\cubrfl{\xi\tr Y > t}}\\
&=
\limsup_{t\toinf}
\robrfl{
\frac{\Prob\cubrfl{\xi\tr X > t}}{\Prob\cubrfl{\norm{X}_1>t}}
\cdot
\frac{\Prob\cubrfl{\norm{Y}_1>t}}{\Prob\cubrfl{\xi\tr Y > t}} 
\cdot
\frac{\Prob\cubrfl{\norm{X}_1 > t}}{\Prob\cubrfl{\norm{Y}_1>t}}
}\\
&=
\frac{\gammaxi(X)}{\gammaxi(Y)}
\cdot 
\limsup_{t\toinf}\frac{\Prob\cubrfl{\norm{X}_1 > t}}{\Prob\cubrfl{\norm{Y}_1>t}}\\
&=0
\ldotp
\end{align*}
\item %
Recall that multivariate regular variation of $X$ implies regular variation of $\norm{X}_1$ with tail index $\alphaX$. Analogously, $\norm{Y}_1$ is regularly varying with tail index $\alphaY$. Finally, $\alphaX > \alphaY$ yields \eqref{eq:3.5} and by~(\ref{item:prop.3.1a}) we obtain $X\orderapl Y$.
\qed
\end{enumerate}
\end{myproof}
\par
Thus the primary setting for studying the influence of dependence 
structures on the ordering of extreme portfolio losses is the case of 
random variables $X$ and $Y$ with equal tail indices:
\[
\alphaX=\alphaY=:\alpha
\ldotp
\] 
In the framework of multivariate regular variation, asymptotic dependence in the tail region 
is characterized by the spectral measure $\Psi$  or its canonical version $\Psiast$.  
The \emph{canonical exponent measure} $\nuast$ of $X$ is obtained from the exponent 
measure $\nu$ as 
\[
\nuast=\nu\circ T
\]
with the transformation $T:\Rd\to\Rd$ defined by
\begin{equation}
\quad
T(x)
:=\label{eq:3.7}
\robrfl{T_\alpha\robrfl{\nu\robr{B_1}\cdot x\compone},\ldots, T_\alpha\robrfl{\nu\robr{B_d}\cdot x\compd}},
\end{equation}
where 
\begin{equation} 
T_\alpha(t)
:=\label{eq:3.8}
\robrfl{t\pospart^{1/\alpha} - t\negpart^{1/\alpha}}  \text{ and }
B_i := \cubrfl{x\in\Rd: \absfl{x\compi} > 1}
\ldotp
\end{equation}
Furthermore, $\nuast$ exhibits the scaling property 
\[
\nuast(tA)=t\inv\nuast(A),
\quad t>0,
\] 
and, analogously to~\eqref{eq:28}, has a product structure in polar 
coordinates:
\begin{equation}\label{eq:apl.3}
\nuast\circ\tau\inv = \rho_1 \otimes \Psiast,
\end{equation}
The measure $\Psiast$ is the \emph{canonical spectral measure} of $X$. 
\par  
Since $\orderapl$ and $\Psiast$ are invariant under componentwise rescalings, 
the canonical spectral measure $\Psiast$ is more suitable for the 
characterization of $\orderapl$. 
The following lemma provides a representation of the extreme risk index 
$\gammaxi$ in terms of $\Psiast$.  Note that the formulation makes use of 
the componentwise product notation~\eqref{eq:apl.2}.
\par
\begin{proposition}
\label{prop:3.2}
Let $X$ be multivariate regularly varying on $\Rd$ with tail index 
$\alpha\in(0,\infty)$. 
If $X$ satisfies the non-degeneracy condition~\eqref{eq:4},  
then 
\begin{equation}\label{eq:3.9}
\gammaxi(X)=\int_{\Sdone}\gxialpha\robrfl{v s} \, \dm \Psiast(s),
\end{equation}
where $\Psiast$ denotes the canonical spectral measure of $X$, 
the rescaling vector $v=\robr{v\compone,\ldots,v\compd}$ 
is defined by 
\begin{equation}\label{eq:3.10}
v\compi:=\robr{\gammaei(X)+\gammaminusei(X)},
\end{equation}  
and the function $\gxialpha:\Rd\to\R$ is defined as 
\begin{equation}\label{eq:3.11}
\gxialpha(x):=%
\robrfl{\sum_{i=1}^d\xi\compi\cdot\robrfl{\robrfl{x\compi}\pospart^{1/\alpha} - \robrfl{x\compi}\negpart^{1/\alpha}}}\pospart\powalpha
\ldotp
\end{equation}
\end{proposition}
\par
\begin{myproof}
Denote $\Axione:=\{x\in\Rd: \xi\tr x\ge 1\}$. Then, by definition of $\nuast$,
\begin{align}
\gammaxi(X)
&=\nonumber
\nu(\Axione)\\
&=\nonumber
\nuast\robr{T\inv(\Axione)}\\
&=\nonumber
\nuast\cubrfl{x\in\Rd: T(x) \in \Axione}\\
&=\label{eq:3.12}
\int_{\Sdone}\int_{(0,\infty)} 
1\cubrfl{\xi\tr T(rs) > 1} \, \dm \rho_1(r) \, \dm \Psiast(s)
\ldotp
\end{align}
It is easy to see that~\eqref{eq:3.8} implies 
$T_\alpha(rt)=r^{1/\alpha}T_\alpha(t)$ for $r>0$ and $t\in\R$. 
Consequently, \eqref{eq:3.7} yields
\begin{equation}
\label{eq:3.13}
T(rx) = r^{1/\alpha} T(x)
\end{equation}
for $r>0$ and $x\in\Rd$. 
Applying~\eqref{eq:3.13} to~\eqref{eq:3.12}, one obtains  
\begin{align}
\gammaxi(X)
&=\nonumber
\int_{\Sdone}\int_{(0,\infty)} 
1\cubrfl{r^{1/\alpha} \xi\tr T(s) > 1} \, \dm \rho_1(r) \, \dm \Psiast(s)\\
&=\nonumber
\int_{\Sdone}\int_{(0,\infty)} 
1\cubrfl{\xi\tr T(s) > 0}
1\cubrfl{r > \robrfl{\xi\tr T(s)}\powminusalpha} \, \dm \rho_1(r) \, \dm \Psiast(s)\\
&=\nonumber
\int_{\Sdone}
1\cubrfl{\xi\tr T(s) > 0} \robrfl{\xi\tr T(s)}\powalpha \, \dm \Psiast(s)\\
&=\label{eq:3.14}
\int_{\Sdone}
\robrfl{\xi\tr T(s)}\powalpha\pospart \, \dm \Psiast(s)
\ldotp
\end{align}
Finally, consider the sets $B_i$ defined in~\eqref{eq:3.8}. 
It is easy to see that 
\begin{equation*}
\nu(B_i) 
= 
\gammaei(X)+\gammaminusei(X) = v\compi
\ldotp 
\end{equation*}
Hence 
\begin{align*}
\robrfl{\xi\tr T(s)}\pospart\powalpha 
&= 
\robrfl{
\sum_{i=1}^d\xi\compi\cdot\robrfl{T_\alpha\robrfl{v\compi s\compi}}
}  \pospart\powalpha\\
&=
\gxialpha\robrfl{vs}
\ldotp
\qed
\end{align*}
\end{myproof}

As already mentioned above, $\orderapl$ and $\Psiast$ are invariant under 
rescaling of components. 
Consequently, characterization of $\orderapl$ can be reduced to the case when the marginal weights $v\compi=\gammaei(X)+\gammaminusei(X)$ in~\eqref{eq:3.9} are standardized by
\begin{equation}
\label{eq:3.15}
\forall i,j\in\cubr{1,\ldots,d}
\quad 
\lim_{t\toinf} 
\frac{\Prob\cubr{\abs{X\compi}> t}}{\Prob\cubr{\abs{X\compj}>t}}
=1
\ldotp
\end{equation}
This condition will be referred to as the 
\emph{balanced tails condition}. 
The following result shows that this condition significantly simplifies the representation~\eqref{eq:3.9}. 
\par
\begin{proposition}
\label{prop:3.3}
Suppose that $X$ is  multivariate regularly varying on $\Rd$ with tail index
$\alpha\in(0,\infty)$. 
%
\begin{enumerate}[(a)]
\item 
\label{item:L39.3}%
If $X$ has balanced tails in the sense of~\eqref{eq:3.15}, then 
\begin{equation}
\label{eq:3.17}
\frac{\gammaxi(X)}{\gammaeone(X) + \gammaminuseone(X)} = \Psiast \gxialpha
\ldotp
\end{equation}
\item
\label{item:L39.1}%
The non-degeneracy condition~\eqref{eq:4} is equivalent to 
the existence of a vector $w\in(0,\infty)^d$ 
such that $wX$ has balanced tails.  
\item
\label{item:L39.2}%
The extreme risk index $\gammaxi$ of the rescaled vector $wX$ obtained 
in part~(\ref{item:L39.1})  satisfies
\begin{equation}
\label{eq:3.18}
\frac{\gammaxi(wX)}{\gammaeone(wX) + \gammaminuseone(wX)} = \PsiastX \gxialpha
\ldotp
\end{equation}
\end{enumerate}
\end{proposition}
\par
\begin{myproof}
Part~(\ref{item:L39.3}). 
Consider the integrand $\gxialpha(vs)$ in the representation~\eqref{eq:3.9}:
\[
\gxialpha(vs)
=
\robrfl{\sum_{i=1}^{d} \xi\compi \cdot 
\robrfl{\robrfl{v\compi s\compi}\pospart\powonebyalpha - \robrfl{v\compi s\compi}\negpart\powonebyalpha}
}\pospart\powalpha
\ldotp
\]
The balanced tails condition~\eqref{eq:3.15} implies that $X$ is 
non-degenerate in the sense of~\eqref{eq:4}. 
Furthermore, all weights $v\compi$ in the representation~\eqref{eq:3.9} 
are equal: 
\begin{align*}
1
&=
\lim_{t\toinf}
\frac{\Prob\cubrfl{\absfl{X\compi}>t} / \Prob\cubrfl{\normfl{X}_1 >t}} 
{\Prob\cubrfl{\absfl{X\compj}>t} / \Prob\cubrfl{\normfl{X}_1 >t}}
=
\frac{\gammaei(X) + \gammaminusei(X)}{\gammaej(X) + \gammaminusej(X)}\\
&=
\frac{v\compi}{v\compj}
,\quad i,j\in\cubr{1,\ldots,d}
\ldotp
\end{align*}
Hence $\gxialpha(vs)$ simplifies to   
\begin{align*}
\gxialpha(vs)
&=
v\compone \gxialpha(s)\\
&=
\robrfl{\gammaeone(X) + \gammaminuseone(X)} \gxialpha(s)
\ldotp
\end{align*} 
\par
Part~(\ref{item:L39.1}). 
Suppose that $X$ satisfies~\eqref{eq:4}. Then the sets $B_i$ defined in~\eqref{eq:3.8} satisfy $\nu(B_i)>0$ for $i=1,\ldots,d$. Consequently, the random variables $\abs{X\compi}$ are regularly varying with tail index $\alpha$. Denoting 
\begin{equation}
\label{eq:3.19}
w\compi:=\robr{\nu(B_i)}\powminusonebyalpha,
\end{equation}
one obtains 
\begin{align*}
\lim_{t\toinf}\frac{\Prob\cubrfl{\absfl{w\compi X\compi} > t}}{\Prob\cubrfl{\norm{X}_1>t}} 
&=
\lim_{t\toinf}
\robrfl{
\frac{\Prob\cubrfl{\absfl{X\compi} > t/w\compi}}{\Prob\cubrfl{\absfl{X\compi} > t}}
\cdot
\frac{\Prob\cubrfl{\absfl{X\compi} > t}}{\Prob\cubrfl{\norm{X}_1>t}} 
}\\
&= \robrfl{w\compi}\powalpha\cdot\nu(B_i)\\
&= 1
\end{align*}
for $i=1,\ldots,d$. Hence, for any $i,j\in\cubr{1,\ldots,d}$, 
\begin{align*}
\lim_{t\toinf}
\frac{\Prob\cubrfl{\absfl{w\compi X\compi} > t}}{\Prob\cubrfl{\absfl{w\compj X\compj} > t}}
&=
1
\ldotp
\end{align*}
\par
To prove the inverse implication, suppose that $Z:=wX$ has balanced tails 
for some $w\in(0,\infty)^d$. Then the the exponent measure $\nu$ of $X$
satisfies 
\begin{align*}
\frac{\nu(B_i)}{\nu(B_1)}
&=
\lim_{t\toinf}\frac{\Prob\cubrfl{\absfl{X\compi}>t}}{\Prob\cubrfl{\absfl{X\compone} > t}}\\
&=
\lim_{t\toinf}\frac{\Prob\cubrfl{\absfl{Z\compi}>w\compi t}}{\Prob\cubrfl{\absfl{Z\compone} > w\compone t}}\\
&=
\lim_{t\toinf}\robrfl{
\frac{\Prob\cubrfl{\absfl{Z\compi}> w\compi t}}{\Prob\cubrfl{\absfl{Z\compi} >t}}
\cdot
\frac{\Prob\cubrfl{\absfl{Z\compone} >t}}{\Prob\cubrfl{\absfl{Z\compone} > w\compone t}}
\cdot
\frac{\Prob\cubrfl{\absfl{Z\compi} >t}}{\Prob\cubrfl{\absfl{Z\compone} >t}}
}\\
&=
\robrfl{\frac{w\compi}{w\compone}}\powminusalpha
\in(0,\infty)
,\quad i\in\cubrfl{1,\ldots,d}
\ldotp
\end{align*}
Since multivariate regular variation of $X$ implies $\nu(B_j)>0$ for at least 
one index $j\in\cubr{1,\ldots,d}$, this yields $\nu(B_i)>0$ for all $i$.
\par

Part~(\ref{item:L39.2}). This is an immediate consequence of 
part (\ref{item:L39.3}) and the invariance of canonical spectral measures 
under componentwise rescaling.
\end{myproof}
\par
Representation~\eqref{eq:3.17} suggests that ordering of 
the normalized extreme risk indices $\gammaxi/(\gammaeone + \gammaminuseone)$ 
in the balanced tails setting can be considered 
as an \emph{integral order relation} 
for canonical spectral measures with respect to the function class
\begin{equation}\label{eq:3.20}
\Gcalalpha:=\cubrfl{\gxialpha:\xi\in\Simpd}
\ldotp
\end{equation}
This justifies the following definition.
\par
\begin{definition}
\label{def:3.4}
Let $\Psiast$ and $\Phiast$ be canonical spectral measures on $\Sdone$
and let $\alpha>0$. 
Then the order relation $\Psiast \orderGcalalpha \Phiast$ is defined by
\begin{equation}\label{eq:3.21}
\forall g\in\Gcalalpha
\quad
\Psiast g \le\Phiast g
\ldotp
\end{equation}
\end{definition}
\par
\begin{remark}
\label{rem:3.1}
\begin{enumerate}[(a)]
\item
\label{item:r14.1}%
For $\alpha=1$ and spectral measures on $\Simpd$ the extreme risk index 
$\gammaxi(X)$ is linear in $\xi$ \citep[cf.][Lemma~3.2]{Mainik/Rueschendorf:2010}.
Consequently, $\orderGcalalpha$ is indifferent in this case, 
i.e., 
any $\Psiast$ and $\Phiast$ on $\Borel\robr{\Simpd}$ satisfy 
\begin{equation}
\label{eq:3.22}
\Psiast \order_{\Gcal,1} \Phiast 
\quad\text{and}\quad
\Phiast \order_{\Gcal,1} \Psiast
\ldotp
\end{equation}
\item
\label{item:r14.2}%
The order relation $\orderGcalalpha$ is mixing invariant  
in the sense that 
uniform ordering of two parametric families 
$\cubr{\Psiast_\theta:\theta\in\Theta}$ and 
$\cubr{\Phiast_\theta:\theta \in\Theta}$,  
\[
\forall \theta\in\Theta
\quad
\Psiast_\theta\orderGcalalpha\Phiast_\theta
,
\]
implies 
\[
\int_\Theta\Psiast_\theta \, \dm \mu(\theta) 
\orderGcalalpha 
\int_\Theta\Phiast_\theta \, \dm \mu(\theta)
\]
for any probability measure $\mu$ on $\Theta$. 
\end{enumerate}
\end{remark}
\par
%
%
The following theorem states that $\orderapl$ is in a certain sense  
equivalent to the ordering of canonical spectral measures 
and allows to reduce the verification of $\orderapl$ to the verification 
of $\orderGcalalpha$. 
Some exemplary applications are given in  Section~\ref{sec:5}.    
Furthermore, given explicit representations of spectral measures or their 
canonical versions,  
this result allows to verify $\orderapl$ numerically, which is very 
useful in practice. 
\par
\begin{theorem}
\label{theo:3.4}
Let $X$ and $Y$ be multivariate regularly varying random vectors on $\Rd$ with tail index $\alpha\in(0,\infty)$ and canonical  spectral measures $\PsiastX$ and $\PsiastY$. Further, suppose that $X$ and $Y$ satisfy the balanced tails condition~\eqref{eq:3.15}. 
\begin{enumerate}[(a)]
\item
\label{item:t4.1}%
If $\absfl{X\compone} \orderapl \absfl{Y\compone}$,
then $\PsiastX \orderGcalalpha \PsiastY$ implies $X \orderapl Y$.
\vspace{0.5em}
\item
\label{item:t4.2}%
If 
$\absfl{X\compone} \orderapl \absfl{Y\compone}$ 
and 
$\absfl{Y\compone} \orderapl \absfl{X\compone}$,
then 
$\PsiastX \orderGcalalpha \PsiastY$ is equivalent to  $X \orderapl Y$.
\end{enumerate}
\end{theorem}
\par
\begin{myproof}
(\ref{item:t4.1}) 
Since $X$ has balanced tails, Proposition \ref{prop:3.3}(\ref{item:L39.3}) yields
\begin{align*}
\lim_{t\toinf}
\frac{\Prob\cubrfl{\xi\tr X >t}}{\Prob\cubrfl{\absfl{X\compone}>t}}
&=
\lim_{t\toinf}
\robrfl{
\frac{\Prob\cubrfl{\xi\tr X >t}}{\Prob\cubrfl{\norm{X}_1>t}}
\cdot
\frac{\Prob\cubrfl{\norm{X}_1>t}}{\Prob\cubrfl{\absfl{X\compone}>t}}
}\\
&=
\frac{\gammaxi(X)}{\gammaeone(X) + \gammaminuseone(X)}\\
&=
\PsiastX \gxialpha
\ldotp
\end{align*}
Analogously one obtains
\[
\lim_{t\toinf}
\frac{\Prob\cubrfl{\xi\tr Y >t}}{\Prob\cubrfl{\absfl{Y\compone}>t}}
=
\PsiastY \gxialpha
\ldotp
\]
Moreover, $\PsiastX \orderGcalalpha \PsiastY$ implies 
\begin{equation}
\label{eq:3.23}
\frac{\PsiastX\gxialpha}{\PsiastY\gxialpha}
\le
1
\ldotp
\end{equation}
Consequently, 
\begin{align}
\mylefteqn\nonumber
\limsup_{t\toinf}
\frac{\Prob\cubrfl{\xi\tr X >t}}{\Prob\cubrfl{\xi\tr Y >t}}\\
&=\nonumber
\limsup_{t\toinf}\robrfl{
\frac{\Prob\cubrfl{\xi\tr X >t}}{\Prob\cubrfl{\absfl{X\compone}>t}}
\cdot
\frac{\Prob\cubrfl{\absfl{Y\compone}>t}}{\Prob\cubrfl{\xi\tr Y >t}}
\cdot
\frac{\Prob\cubrfl{\absfl{X\compone}>t}}{\Prob\cubrfl{\absfl{Y\compone}>t}}
}\\
&=\label{eq:3.24}
\frac{\PsiastX\gxialpha}{\PsiastY\gxialpha}
\cdot
\limsup_{t\toinf}\frac{\Prob\cubrfl{\absfl{X\compi}>t}}{\Prob\cubrfl{\absfl{Y\compi}>t}}\\
&\le\nonumber
1
\end{align}
due to~\eqref{eq:3.23} and $\abs{X\compi} \orderapl \abs{Y\compi}$.
\par 
\medskip
\noindent %
(\ref{item:t4.2})
By part~(\ref{item:t4.1}), it suffices to show that 
$X \orderapl Y$ implies $\PsiastX \orderGcalalpha \PsiastY$. 
By assumption $\abs{X^{(1)}}$ and $\abs{Y^{(1)}}$ have asymptotically equivalent tails,
\[
\lim_{t\toinf}
\frac{\Prob\cubrfl{\absfl{X\compone} >t}}{\Prob\cubrfl{\absfl{Y\compone} >t}} 
= 
1
\ldotp
\]
Thus~\eqref{eq:3.24} yields
\[
\frac{\PsiastX\gxialpha}{\PsiastY\gxialpha}
=
\limsup_{t\toinf}
\frac{\Prob\cubrfl{\xi\tr X >t}}{\Prob\cubrfl{\xi\tr Y >t}}
\]
and $X \orderapl Y$ implies $\PsiastX \orderGcalalpha \PsiastY$.
\end{myproof}
\par
The following result answers the question for dependence structures 
corresponding to the best and the worst possible diversification effects 
for multivariate regularly varying random vectors in $\Rplusd$.
According to Theorem~\ref{theo:3.4}, it suffices to find the upper and 
the lower elements with respect to $\orderGcalalpha$ in the set 
of all canonical spectral measures on $\Simpd$. 
It turns out that for $\alpha > 1$  
the best diversification effects are obtained in case of asymptotic 
independence, i.e., the $\orderGcalalpha$-maximal element is given by   
\begin{equation}\label{eq:apl.8}
\Psiast_0 := \sumioned \Dirac{\ei},
\end{equation}
whereas the worst diversification effects are obtained in case of 
the asymptotic comonotonicity, represented by 
\begin{equation}\label{eq:apl.9}
\Psiast_1 := d \cdot \Dirac{(1/d,\ldots,1/d)}
\ldotp
\end{equation}
For $\alpha < 1$ the situation is inverse. 
\begin{theorem}\label{thm:3.8}
Let $\Psiast$ be an arbitrary canonical spectral measure on $\Simpd$ and let 
$\Psiast_0$ and $\Psiast_1$ be defined according to~\eqref{eq:apl.8} 
and~\eqref{eq:apl.9}. Then
\begin{enumerate}[(a)]
\item\label{item:thm:3.8.a}
$\Psiast_0  \orderGcalalpha \Psiast \orderGcalalpha \Psiast_1$ 
for $\alpha \ge 1$.
\vspace{0.5em} 
\item\label{item:thm:3.8.b}
$\Psiast_1  \orderGcalalpha \Psiast \orderGcalalpha \Psiast_0$ 
for $\alpha \in (0,1]$. 
\end{enumerate}
\end{theorem}
\begin{proof}
Let $X$ be multivariate regularly varying on $\Rplusd$ with canonical spectral 
measure $\Psiast$. Without loss of generality we can assume that $X$ 
satisfies the balanced tails condition~\eqref{eq:3.15}. 
Then, according to~\eqref{eq:3.17}, we have 
\begin{equation}
\label{eq:apl.6}
\Psiast\gxialpha=\frac{\gammaxi(X)}{\gammaeone(X)}
\ldotp
\end{equation}
Furthermore, 
we have $\Psiast g_{\ei,\alpha} = 1$ for $i=1,\ldots,d$.  
Recall that the mapping $\xi\mapsto\gammaxi$ is convex for $\alpha \ge 1$ 
\citep[cf.][Lemma~3.2]{Mainik/Rueschendorf:2010}. 
Due to~\eqref{eq:apl.6} this behaviour is inherited by the mapping 
$\xi\mapsto\Psiast\gxialpha$. 
Thus for $\alpha \ge 1$  
we have $\Psiast\gxialpha \le 1 = \Psiast_1\gxialpha$ for all $\xi\in\Simpd$, 
which exactly means $\Psiast \orderGcalalpha \Psiast_1$ for $\alpha \ge 1$. 
\par
To complete the proof of part~(\ref{item:thm:3.8.a}), note that 
the normalization of canonical spectral measures yields 
\begin{equation}\label{eq:apl.7}
\forall\xi\in\Simpd
\quad
\Psiast_0\gxialpha 
= 
\sumioned\robrfl{\xi\compi}\powalpha
=
\int_{\Simpd} \sumioned \robrfl{\xi\compi}\powalpha s \compi  
\,\Psiast(\dm s)
\end{equation}
Comparing the integrand on the right side of \eqref{eq:apl.7} with 
the function $\gxialpha(s)=\robr{\xi\tr s\powonebyalpha}\powalpha$, 
we see that 
\[
\sumioned \robrfl{\xi\compi}\powalpha s\compi
=
\gxialpha(s)
\cdot
\sumioned z_i\powalpha
\]
with
\[
z_i := \frac{\xi\compi \cdot\robrfl{s\compi}\powonebyalpha} 
{\xi\tr s\powonebyalpha}
\ldotp
\]
Thus it suffices to demonstrate that $\sumioned z_i\powalpha \le 1$, 
which follows from $z_i\in[0,1]$, $z_i\powalpha \le z_i$ for $\alpha \ge 1$, and 
$\sumioned z_i=1$.
\par
The inverse result for $\alpha\in(0,1]$ stated in~(\ref{item:thm:3.8.b})  
follows from 
the concavity of the mapping $\xi\mapsto\Psiast\gxialpha$ 
and the inequality $z_i\powalpha \ge z_i$.  
\end{proof}
\par
Due to Theorem~\ref{theo:3.4}, an analogue of the foregoing 
result for $\orderapl$ is straightforward.
\begin{corollary}\label{cor:3.10}
Let $X$ be multivariate regularly varying in $\Rplusd$ with tail index 
$\alpha\in(0,\infty)$ and identically 
distributed margins $X\compi\sim F$, $i=1,\ldots,d$.
Further, let $Y$ be a random vector with independent margins 
$Y \compi\sim F$, and let $Z$ be a random vector with totally dependent 
margins $Z\compi=Z\compone$ $\Prob$-a.s.\ and $Z\compone\sim F$. 
Then
\begin{enumerate}[(a)]
\item
$Y \orderapl X \orderapl Z$ for $\alpha \ge 1$
\item
$Z \orderapl X \orderapl Y$ for $\alpha \in (0,1]$.
\end{enumerate}
\end{corollary}
\par
\begin{remark}
The strict assumptions of Corollary~\ref{cor:3.10} are chosen for clearness and simplicity.
The independence of $Y\compi$ and the total dependence of $Z\compi$ 
are needed only in the tail region, i.e., it suffices for $Y$ and $Z$ 
to be multivariate regularly varying with canonical spectral measures 
$\Psiast_0$ and $\Psiast_1$, respectively.
Furthermore, the assumption of identically distributed margins 
can be replaced by equivalent tails:
\[
1=
\lim_{t\toinf}\frac{\Prob\cubrfl{Y\compi >t}}{\Prob\cubrfl{X\compi >t}}
=
\lim_{t\toinf}\frac{\Prob\cubrfl{Z\compi >t}}{\Prob\cubrfl{X\compi >t}}
,\quad i=1,\ldots,d
\ldotp
\]
Finally, the non-negativity of $X\compi$, $Y\compi$, and $Z\compi$  
is needed only in the asymptotic sense. The ordering results remain true 
if the spectral measure of $X$ is restricted to the unit simplex $\Simpd$.  
\end{remark}
\par
Combining Theorem~\ref{theo:3.4} with Theorem~\ref{theo:2.4}, one obtains an 
ordering result for the canonical spectral measures of multivariate regularly 
varying elliptical distributions. 
The notation $\Psiast=\Psiast(\alpha,C)$ is justified by the fact that 
spectral measures of elliptical distributions depend only on  the tail 
index $\alpha$ and the generalized covariance matrix $C$. 
An explicit representation of spectral densities for bivariate elliptical 
distributions was obtained by \citet{Hult/Lindskog:2002}. 
Alternative representations that are valid for all dimensions $d\ge2$ 
are given in \citet{Mainik:2010}, Lemma~2.8.
\par
\begin{proposition}
\label{prop:3.7}
Let $C$ and $D$ be $d$-dimensional covariance matrices satisfying 
\begin{equation}
\label{eq:3.25}
\Cii = \Dii > 0
,\quad
i = 1,\ldots,d,
\end{equation}
and 
\begin{equation}
\label{eq:3.26}
\forall \xi\in\Simpd 
\quad
\xi\tr C \xi \le \xi\tr D \xi
\ldotp
\end{equation}
Then 
\[
\forall \alpha>0 
\quad
\Psiast\robr{\alpha,C} \orderGcalalpha \Psiast\robr{\alpha,D}
\ldotp
\]
\end{proposition}
\par
\begin{myproof}
Fix $\alpha\in(0,\infty)$ and consider random vectors 
\[
X\disteq R A U
,\quad
Y\disteq R  B  U, 
\]
where $A$ and $B$ are square roots of the matrices $C$ and $D$ 
in~\eqref{eq:3.26}, i.e., 
\[ 
C=A  A\tr
,\quad 
D=B B\tr,
\]
and $R$ is an arbitrary regularly varying non-negative 
random variable with tail index $\alpha$.
\par
As a consequence of Theorem \ref{theo:2.4} one obtains $X \orderapl Y$. 
Furthermore, invariance of $\orderapl$ under componentwise rescaling 
yields $wX \orderapl wY$ for $w=\robr{w\compone,\ldots,w\compd}$ with   
\[
w\compi:={\Cii}^{-1/2}={\Dii}^{-1/2}
,\quad 
i=1,\ldots,d
\ldotp
\]
Moreover, as a particular consequence of arguments 
underlying \eqref{eq:2.17}, one obtains 
\[
w\compi X\compi \disteq w\compj Y\compj
,\quad
i,j\in\cubr{1,\ldots,d}
\ldotp
\]
Hence the random vectors $wX$ and $wY$ satisfy the balanced tails condition \eqref{eq:3.15}, whereas their components are mutually ordered with respect to $\orderapl$.
Finally, Theorem~\ref{theo:3.4}(\ref{item:t4.2}) and invariance of canonical spectral measures under componentwise rescalings yield
\begin{equation*}
\Psiast(\alpha,C) = \Psiast_{wX} 
\orderGcalalpha 
\Psiast_{wY} = \Psiast(\alpha,D)
\ldotp
%
\qed
%
\end{equation*}
\end{myproof}
\par
The subsequent result extends Theorem~\ref{theo:3.4} to random vectors that do not have balanced tails.
\par
\begin{theorem}
\label{theo:3.8}
Let $X$ and $Y$ be multivariate regularly varying random vectors on $\Rd$ 
with tail index $\alpha\in(0,\infty)$  
and canonical spectral measures $\PsiastX$ and $\PsiastY$. 
Further, assume that $\abs{X\compi} \orderapl \abs{Y\compi}$ with 
\begin{equation}
\label{eq:3.27}
\lambda_i
:=
\lim_{t\toinf}
\frac{\Prob\cubrfl{\absfl{X\compi}>t}}{\Prob\cubrfl{\absfl{Y\compi}>t}} 
\in(0,1]
\end{equation}
for $i=1,\ldots,d$  and that the vector $v=\robr{v\compone,\ldots,v\compd}$ defined by
\begin{equation}
\label{eq:3.28}
v\compi:=\lambda_i^{-1/\alpha}
\end{equation} 
satisfies 
\begin{equation}
\label{eq:3.29}
X \orderapl v X
\quad\text{or}\quad
v\inv Y \orderapl Y
\ldotp
\end{equation}
Then $\PsiastX \orderGcalalpha \PsiastY$ implies $X \orderapl Y$.
\end{theorem}
\par
\begin{myproof}
According to Proposition~\ref{prop:3.3}(\ref{item:L39.1}), there exists $w\in\Rplusd$ 
such that $wY$ satisfies the balanced tails condition~\eqref{eq:3.15}. 
Furthermore, the tails of the random vector 
\[
vwX:=\robrfl{v\compone w\compone X\compone ,\ldots, v\compd w\compd X\compd}
\]
with $v$ defined in~\eqref{eq:3.27} are also balanced. Indeed, it is easy to see that 
\[
\lim_{t\toinf}
\frac{\Prob\cubrfl{\absfl{w\compi Y\compi} >t}}{\Prob\cubrfl{\absfl{Y\compi} >t}} 
=
\lim_{t\toinf}
\frac{\Prob\cubrfl{\absfl{v\compi w\compi X\compi} >t}}{\Prob\cubrfl{\absfl{v\compi X\compi} >t}} 
=
\robrfl{w\compi}\powalpha
\]
for $i=1,\ldots,d$. Analogously one obtains 
\[
\lim_{t\toinf}
\frac{\Prob\cubrfl{\absfl{v\compi X\compi} >t}}{\Prob\cubrfl{\absfl{X\compi} >t}} 
=
\robrfl{v\compi}\powalpha = \lambda_i\inv
\]
and, as a result,
\begin{align*}
\mylefteqn
\lim_{t\toinf} 
\frac{\Prob\cubrfl{\absfl{v\compi w\compi X\compi} >t}}{\Prob\cubrfl{\absfl{w\compi Y\compi} > t}}\\
&=
\lim_{t\toinf} 
\frac{\Prob\cubrfl{\absfl{v\compi X\compi} >t}}{\Prob\cubrfl{\absfl{Y\compi} > t}}\\
&=
\lim_{t\toinf}\robrfl{
\frac{\Prob\cubrfl{\absfl{v\compi X\compi} >t}}{\Prob\cubrfl{\absfl{X\compi} > t}}
\cdot
\frac{\Prob\cubrfl{\absfl{X\compi} > t}}{\Prob\cubrfl{\absfl{Y\compi} > t}}
}\\
&=
\lambda_i\inv 
\cdot
\lim_{t\toinf}\frac{\Prob\cubrfl{\absfl{X\compi} > t}}{\Prob\cubrfl{\absfl{Y\compi} > t}}\\
&=
1
\end{align*}
for $i=1,\ldots,d$. Hence the balanced tails condition for $wY$ implies that the tails of $vwX$ are also balanced.
\par
Furthermore, invariance of canonical spectral measures under 
componentwise rescaling yields
\[
\Psiast_{vwX} = \PsiastX \orderGcalalpha \PsiastY = \Psiast_{wY}
\ldotp
\]
Thus, applying Theorem~\ref{theo:3.4}(\ref{item:t4.1}), one obtains
\begin{equation}
\label{eq:3.30}
vwX \orderapl wY
\ldotp
\end{equation}
Since $v\compi=\lambda_i\powminusonebyalpha > 0$ for $i=1,\ldots,d$, 
condition~\eqref{eq:3.30} is equivalent to 
\begin{equation}
\label{eq:3.31}
wX \orderapl v\inv wY
\ldotp
\end{equation}
Moreover, assumption~\eqref{eq:3.29} implies
\begin{equation}
\label{eq:3.32}
wX \orderapl vwX
\quad\text{or}\quad 
v^{-1}wY\orderapl wY.
\end{equation}
%
%
Combining this ordering statement 
with \eqref{eq:3.30} and~\eqref{eq:3.31}, 
one obtains 
\[
wX \orderapl wY
\ldotp
\]
Finally, 
invariance of $\orderapl$ with respect to componentwise rescaling yields 
$X \orderapl Y$.
\end{myproof}
\par
In the special case of random vectors in $\Rplusd$ Theorem~\ref{theo:3.8} 
can be simplified to the following result. 
\par
\begin{corollary}
\label{cor:3.9}  
Let $X$ and $Y$ be multivariate regularly varying random vectors on $\Rplusd$ 
with tail index $\alpha\in(0,\infty)$ and canonical 
spectral measures $\PsiastX$ and $\PsiastY$. 
Further, suppose that 
\begin{equation}
\label{eq:3.33}
\lambda_i
:=
\limsup_{t\toinf}
\frac{\Prob\cubrfl{\absfl{X\compi}>t}}{\Prob\cubrfl{\absfl{Y\compi}>t}} \in (0,1], 
\quad
i = 1,\ldots,d
\ldotp
\end{equation}
Then $\PsiastX \orderGcalalpha \PsiastY$ implies $X \orderapl Y$.
\end{corollary}
\par
\begin{myproof}
Assumption~\eqref{eq:3.33} yields that the rescaling vector $v$ 
defined in \eqref{eq:3.28} is an element of $[1,\infty)^d$. 
Thus $v-(1,\ldots,1)\in\Rplusd$ and, since $X$ takes values in $\Rplusd$, 
we have 
\[
X \orderapl X + \robr{v - (1,\ldots,1)}X = vX
\ldotp
\]
Similar arguments yield $v\inv Y\orderapl Y$. 
Hence condition~\eqref{eq:3.29} of Theorem~\ref{theo:3.8} is satisfied.
\end{myproof}
\par
The final result of this section is due to the indifference of 
$\orderGcalalpha$ for $\alpha=1$ mentioned  in 
Remark~\ref{rem:3.1}(\ref{item:r14.1}).
This special property of spectral measures on $\Simpd$ allows to reduce 
$\orderapl$ to the ordering of components. 
It should be noted that this result cannot be extended to the 
general case of spectral measures on $\Sdone$. 
\par
\begin{lemma}
\label{lem:3.10}
Let $X$ and $Y$ be multivariate regularly varying on $\Rplusd$ with tail 
index $\alpha=1$. 
Further, suppose that $Y$ satisfies the non-degeneracy 
condition~\eqref{eq:4} and that 
$X\compi \orderapl Y\compi$ for $i=1,\ldots,d$. Then $X \orderapl Y$.   
\end{lemma}
\par
\begin{myproof}
According to Proposition \ref{prop:3.3}(\ref{item:L39.1}), 
there exists $w\in(0,\infty)^d$ such that $wY$ satisfies the balanced tails 
condition~\eqref{eq:3.15}. 
Furthermore, due to the invariance of $\orderapl$ under componentwise 
rescaling, $X \orderapl Y$ is equivalent to $wX \orderapl wY$.  
\par
Thus it can be assumed without loss of generality that $Y$ has balanced tails. 
This yields 
\[
\lambda_i
:=
\limsup_{t\toinf}\frac{\Prob\cubrfl{X\compi >t}}{\Prob\cubrfl{Y\compi >t}}
=
\limsup_{t\toinf}\frac{\Prob\cubrfl{X\compi >t}}{\Prob\cubrfl{Y\compone >t}}
,\quad
i=1,\ldots,d
\ldotp
\]
Hence the  assumption $X\compi \orderapl Y\compi$ for $i=1,\ldots,d$ 
implies $\lambda_i\in[0,1]$ for all $i$.
Moreover, the balanced tails condition for $Y$ yields
\begin{equation}
\label{eq:3.34}
\gammaeone(Y)=\ldots=\gammaed(Y)
\ldotp
\end{equation}
\par

Now consider the random vector $X$ and denote 
\[
j:=\argmax_{i\in\cubr{1,\ldots,d}} \gammaei(X)
\ldotp
\]
Recall that $\gammaei(X)=\nu_X\robr{\cubr{x\in\Rplusd: x\compi>1}}$ 
with $\nu_X$ 
denoting the exponent 
measure of $X$ and that $\nu_X$ is non-zero. This yields $\gammaej(X)>0$ 
even if $X$ does not satisfy the non-degeneracy condition~\eqref{eq:4}. 
Moreover, for $\alpha=1$, the mapping 
$\xi\mapsto\gammaxi(X)$ is linear. This implies
\begin{equation}
\label{eq:3.35}
\gammaxi(X)
=
\sumioned \xi\compi \cdot  \gammaei(X)
\le 
\gammaej(X)
,\quad 
\xi\in\Simpd
\end{equation}
%
and \eqref{eq:3.34} yields
\begin{equation}
\label{eq:3.36}
\gammaxi(Y) = \sumioned \xi\compi \cdot \gammaei(Y) = \gammaeone(Y)
,\quad
\xi\in\Simpd
\ldotp
\end{equation}
Hence
\begin{align*}
\mylefteqn
\limsup_{t\toinf}
\frac{\Prob\cubrfl{\xi\tr X > t}}{\Prob\cubrfl{\xi\tr Y >t}}\\
&=
\limsup_{t\toinf}\robrfl{
\frac{\Prob\cubrfl{\xi\tr X > t}}{\Prob\cubrfl{X\compj > t}}
\cdot
\frac{\Prob\cubrfl{X\compj > t}}{\Prob\cubrfl{Y\compone > t}} 
\cdot
\frac{\Prob{\cubrfl{Y\compone > t}}}{\Prob\cubrfl{\xi\tr Y >t}}
}\\
&=
\frac{\gammaxi(X)}{\gammaej(X)}
\cdot
\lambda_j
\cdot
\frac{\gammaeone(Y)}{\gammaxi(Y)} \\
&\le
1
\end{align*}
due to $\lambda_j\le 1$, \eqref{eq:3.35}, and \eqref{eq:3.36}.
\end{myproof}
%
%
%
%
%
\section{Relations to convex and supermodular orders}\label{sec:4}
As mentioned in Remark~\ref{rem:apl.1}(\ref{item:apl.2}), 
dependence orders $\ordersm$, $\orderdcx$ 
and convexity orders $\ordercx$, $\ordericx$, $\orderplcx$ 
do not imply $\orderapl$ in general. 
However, it turns out that the relationship between $\orderapl$ and
the ordering of canonical spectral measures by $\orderGcalalpha$ allows 
to draw conclusions of this type  
in the special case of multivariate regularly varying models.  
The core result of this section is stated in Theorem~\ref{thm:5}. 
It entails 
a collection of sufficient 
criteria for $\orderapl$ in terms of convex and supermodular order relations,   
with particular interest paid to the 
inversion of diversification effects for $\alpha<1$.   
An application to copula based models is given in Proposition~\ref{prop:4.4}.
%
\par 
This approach was applied by \citet{Embrechts/Neslehova/Wuethrich:2009} 
to the ordering of risks for the portfolio vector 
$\xi=(1,\ldots,1)$ and for a specific family of multivariate 
regularly varying models with identically distributed, non-negative margins 
$X\compi$ (cf. Example~\ref{ex:2} in Section~\ref{sec:5}). 
\par
The next theorem is the core element of this section. 
It generalizes the arguments of \citet{Embrechts/Neslehova/Wuethrich:2009} 
to multivariate regularly varying random vectors in $\Rd$ with balanced 
tails and tail index $\alpha\ne 1$. 
The case $\alpha=1$ is not included for two reasons.
First, this case is partly trivial due to the indifference of $\orderGcalalpha$ 
for spectral measures on $\Simpd$ (cf.\ Remark~\ref{rem:3.1}(\ref{item:r14.1})).
Second, Karamata's theorem  used in the proof of the integrable case $\alpha>1$ does not yield the desired result for random variables with tail index $\alpha=1$. 
\par
\begin{theorem}
\label{thm:5}
Let $X$ and $Y$ be multivariate regularly varying on $\Rd$ with identical  
tail index $\alpha\ne 1$. Further, assume that $X$ and $Y$ satisfy 
the balanced tails condition~\eqref{eq:3.15}.  
\begin{enumerate}[(a)]
\item
\label{item:t5.1}
For $\alpha>1$ let
\begin{equation}
\label{eq:309}
\limsup_{t\toinf} 
\frac{\Prob\cubrfl{\absfl{X\compone}>t}}{\Prob\cubrfl{\absfl{Y\compone}>t}}
= 1
\end{equation}
and let there exist $u_0>0$ such that with $\hu(t):=\robr{t-u}\pospart$ 
\begin{equation}
\label{eq:292a}
\forall u \ge u_0 \ \forall \xi\in\Simpd 
\quad
\E  \hu\robrfl{\xi\tr X} 
\le 
\E  \hu\robrfl{\xi\tr Y}
\ldotp 
\end{equation}
Then $\PsiastX \orderGcalalpha \PsiastY$. 
\vspace{0.5em} 
\item
\label{item:t5.2}
For $\alpha<1$ suppose that $\abs{X\compone}$ and $\abs{Y\compone}$ are 
equivalent with respect to $\orderapl$, i.e.,
\begin{equation}
\label{eq:310a}
\absfl{X\compone} \orderapl \absfl{Y\compone}
\quad\text{and}\quad 
\absfl{Y\compone} \orderapl \absfl{X\compone}, 
\end{equation}
and let there exist $\uzer >0$ such that with $\fu(t):=-(t \minbin u)$, 
\begin{equation}
\label{eq:292b}
\forall u \ge \uzer \ \forall \xi\in\Simpd 
\quad
\E \fu\robrfl{\robrfl{\xi\tr X}\pospart} 
\le 
\E \fu \robrfl{\robrfl{\xi\tr Y}\pospart}
\ldotp 
\end{equation}
Then $\PsiastY \orderGcalalpha \PsiastX$. 
\end{enumerate}
\end{theorem}
\par
The proof will be given after some conclusions and remarks. 
In particular, it should be noted that the relation between 
$\orderGcalalpha$ and $\orderapl$ established in Theorem~\ref{theo:3.4} 
immediately yields the following result.
\begin{corollary}
\label{cor:8}
\begin{enumerate}[(a)]  
\item 
\label{item:c8.1}%
If random vectors $X$ and $Y$ satisfy conditions of Theorem~\ref{thm:5}(\ref{item:t5.1}), 
then $X \orderapl Y$;
\item
\label{item:c8.2}%
If $X$ and $Y$ satisfy conditions of Theorem~\ref{thm:5}(\ref{item:t5.2}), 
then $Y \orderapl X$.
\end{enumerate} 
\end{corollary}
\par
It should also be noted that 
conditions \eqref{eq:292a} and \eqref{eq:292b} are asymptotic forms of 
the increasing convex ordering $\xi\tr X \ordericx \xi\tr Y$ 
and the decreasing convex ordering $\xi\tr X \orderdecx \xi\tr Y$, 
respectively. 
The consequences can be outlined as follows.  
\begin{remark}
\label{rem:12}
\begin{enumerate}[(a)]
\item 
The following criteria are sufficient for~\eqref{eq:292a} 
and \eqref{eq:292b} to hold: 
\begin{enumerate}[(i)]
\vspace{0.5em}
\item 
$\robr{\xi\tr X}\pospart \ordercx \robr{\xi\tr Y}\pospart$ 
for all $\xi\in\Simpd$,
\item
$X$ and $Y$ are restricted to $\Rplusd$ and 
$X \order Y$ with $\order$ denoting either 
$\orderplcx$, $\orderlcx$,  $\ordercx$, $\orderdcx$, or $\ordersm$.
\end{enumerate}
\vspace{0.5em} 
\item
Additionally, condition~\eqref{eq:292a} follows from  
$X \order Y$ with $\order$ denoting either 
$\orderplcx$, $\orderlcx$,  $\ordercx$, $\orderdcx$, or $\ordersm$.
\end{enumerate}
\end{remark}
Finally, a comment should be made upon   
convex ordering of non-integrable random variables and  
diversification for $\alpha<1$.
The so-called \emph{phase change} at $\alpha=1$, 
i.e., the inversion of diversification effects 
taking place when the tail index $\alpha$ crosses this critical value, 
demonstrates that the implications of convex ordering are essentially different 
for integrable and non-integrable random variables.
Indeed, it is easy to see that if a random variable $Z$ on $\R$ 
satisfies $\E \sqbr{Z\pospart} = \E \sqbr{Z\negpart}=\infty$,
then the only integrable convex functions of $Z$ are the constant ones.
Moreover, if $Z$ is restricted to $\Rplus$ and $\E Z =\infty$,
then any integrable convex function of $Z$ is necessarily non-increasing. 
%
%
\vspace{0.5em}
\par
\begin{myproofx}{of Theorem~\ref{thm:5}}(\ref{item:t5.1})
Consider the expectations in~\eqref{eq:292a}.
It is easy to see that for $u>0$ 
\begin{align*}
\oneby{u}\E \hu\robrfl{\xi\tr X} 
&=
\oneby{u}\int_{(u,\infty)}\Prob\cubrfl{\xi\tr X > t} \dm t\\
&=
\int_{(1,\infty)} \Prob\cubrfl{\xi\tr X > tu} \dm t
\end{align*}
and, as a consequence, 
\[
\frac{u\inv\E \hu\robrfl{\xi\tr X}}{\Prob\cubrfl{\absfl{X\compone}>u}}
=
\frac{\Prob\cubrfl{\xi\tr X > u}}{\Prob\cubrfl{\absfl{X\compone}>u}}
\int_{(1,\infty)} 
\frac{\Prob\cubrfl{\xi\tr X > tu}}{\Prob\cubrfl{\xi\tr X > u}}
\, \dm t
\ldotp
\]
Moreover, Proposition~\ref{prop:3.3}(\ref{item:L39.3}) implies
\begin{equation}
\label{eq:313}
\lim_{u\toinf} 
\frac{\Prob\cubrfl{\xi\tr X > u}}{\Prob\cubrfl{\absfl{X\compone}>u}}
=
\frac{\gammaxi(X)}{\gammaeone(X)+\gammaminuseone(X)} 
=
\PsiastX\gxialpha
\end{equation}
and Karamata's theorem 
\citep[cf.][Theorem B.1.5]{de_Haan/Ferreira:2006}  
yields
\[
\lim_{u\toinf} 
\int_{(1,\infty)} 
\frac{\Prob\cubrfl{\xi\tr X > tu}}{\Prob\cubrfl{\xi\tr X > u}}
\, \dm t
=
\int_{(1,\infty)}
t\powminusalpha \dm t
=
\oneby{\alpha-1}
\ldotp
\]
As a result one obtains
\begin{equation*}
\lim_{u\toinf}
\frac{u\inv \E \hu\robrfl{\xi\tr X}}{\Prob\cubrfl{\absfl{X\compone}>u}}
=
\oneby{\alpha-1}\PsiastX \gxialpha
\end{equation*}
and, analogously,
\begin{equation*}
\lim_{u\toinf}
\frac{u\inv\E \hu\robrfl{\xi\tr Y}}{\Prob\cubrfl{\absfl{Y\compone}>u}}
=
\oneby{\alpha-1}\PsiastY \gxialpha
\ldotp
\end{equation*}
Hence \eqref{eq:292a} and \eqref{eq:309} yield
\begin{align*}
1
&\ge
\limsup_{u\toinf}\frac{u\inv\E\hu\robrfl{\xi\tr X}}{u\inv\E \hu\robrfl{\xi\tr Y}}\\
&=
\limsup_{u\toinf}\robrfl{
\frac{u\inv \E\hu\robrfl{\xi\tr X}}{\Prob\cubrfl{\absfl{X\compone}>u}}
\cdot
\frac{\Prob\cubrfl{\absfl{Y\compone}>u}}{u\inv \E \hu\robrfl{\xi\tr Y}}
\cdot
\frac{\Prob\cubrfl{\absfl{X\compone}>u}}{\Prob\cubrfl{\absfl{Y\compone}>u}}
}\\
&=
\frac{\PsiastX\gxialpha}{\PsiastY\gxialpha}
\end{align*}
for all $\xi\in\Simpd$, which exactly means 
$\PsiastX \orderGcalalpha \PsiastY$. 
\par
\medskip
(\ref{item:t5.2}) 
Note that~\eqref{eq:310a} implies 
\begin{equation}
\label{eq:310}
\lim_{t\toinf} 
\frac{\Prob\cubrfl{\absfl{X\compone}>t}}{\Prob\cubrfl{\absfl{Y\compone}>t}}
= 1
\end{equation}
and that~\eqref{eq:292b} yields
\begin{equation}
\label{eq:311}
\forall u>\uzer \ \forall v\ge 0
\quad
\E \fuplusv\robrfl{\xi\tr X}  - \E \fuplusv\robrfl{\xi\tr Y}  \le 0
\ldotp
\end{equation}
Furthermore, it is easy to see that any random variable $Z$ in $\Rplus$ 
satisfies
\begin{align*}
\E\sqbrfl{Z\minbin u} 
&= 
\int_{(0,\infty)} \robrfl{t \minbin u} \dm \Prob^Z(t)\\
&=
\int_{(0,\infty)}\int_{(0,\infty)} 1\cubr{s<t} \cdot 1\cubr{s<u} \, \dm s \, \dm \Prob^Z(t)\\
&=
\int_{(0,\infty)}1\cubr{s<u} \int_{(0,\infty)} 1\cubr{s<t} \, \dm \Prob^Z(t) \, \dm s\\
&=
\int_{(0,u)}\Prob\cubrfl{Z>s} \dm s
\ldotp
\end{align*}
This implies
\[
\E \fuplusv(Z) = \E \fu(Z) - \int_{(u,u+v)} \Prob\cubrfl{Z>t} \dm t
\ldotp
\]
Consequently, \eqref{eq:311} yields
\begin{equation}
\label{eq:312}
\forall u\ge \uzer\ \forall v>0
\quad 
\E \fu\robrfl{\xi\tr X} - \E\fu\robrfl{\xi\tr Y} 
\le 
I(u,v)
\end{equation}
where 
\begin{align*}
I(u,v)
&:=
\int_{(u, u+v)}
\robrfl{\Prob\cubrfl{\xi\tr X >t } - \Prob\cubrfl{\xi\tr Y > t}} 
\, \dm t\\
&\phantom{:}=
\int_{(u,u+v)} 
\phi(t) \cdot \Prob\cubrfl{\absfl{X\compone}>t}
\,\dm t
\end{align*}
with  
\[
\phi(t)
:=
\frac{\Prob\cubrfl{\xi\tr X > t} - \Prob\cubrfl{\xi\tr Y > t}}
{\Prob\cubrfl{\absfl{X\compone}>t}}
\ldotp
\]
Moreover, \eqref{eq:310}, \eqref{eq:313}, and an 
analogue of~\eqref{eq:313} for $Y$ yield 
\begin{align}
\phi(t)
&=\nonumber
\frac{\Prob\cubrfl{\xi\tr X > t}}{\Prob\cubrfl{\absfl{X\compone}>t}}
-
\frac{\Prob\cubrfl{\xi\tr Y > t}}{\Prob\cubrfl{\absfl{Y\compone}>t}}
\cdot
\frac{\Prob\cubrfl{\absfl{Y\compone}>t}}{\Prob\cubrfl{\absfl{X\compone}>t}}\\
&\to\label{eq:314}
\PsiastX\gxialpha - \PsiastY\gxialpha
,\quad
t\toinf
\ldotp
\end{align}
\par
Now suppose that $\PsiastY \orderGcalalpha \PsiastX$ is not satisfied, 
i.e., there exists $\xi\in\Simpd$ such that 
$\PsiastY\gxialpha > \PsiastX\gxialpha$.
Then~\eqref{eq:314} yields $\phi(t) \le - \eps$ 
for some $\eps>0$ and sufficiently large $t$. 
This implies  
\begin{equation} 
\label{eq:315}
I(u,v)
\le 
-\eps 
\int_{(u, u+v)}
\Prob\cubrfl{\absfl{X\compone}> t} \, \dm t
\end{equation}
for sufficiently large $u$ and all $v\ge0$. 
Moreover, regular variation of $\absfl{X\compone}$ with tail index 
$\alpha<1$ implies $\E\absfl{X\compone}=\infty$. 
Consequently, the integral on the right side of~\eqref{eq:315} tends to 
infinity for $v\toinf$:
\[
\forall u>0
\quad
\lim_{v\toinf}
\int_{(u, u+v)}
\Prob\cubrfl{\absfl{X\compone}> t} \, \dm t
=
\infty
\ldotp
\]
Hence, choosing $u$ and $v$ sufficiently large, one can achieve 
$I(u,v)<c$ for any $c\in\R$. In particular, $u$ and $v$ can be 
chosen such that   
\[
I(u,v) < \E \fu\robrfl{\xi\tr X} - \E\fu\robrfl{\xi\tr Y}, 
\]
which contradicts~\eqref{eq:312}. Thus  
$\PsiastY\gxialpha > \PsiastX\gxialpha$ cannot be true 
and therefore it necessarily holds that $\PsiastY \orderGcalalpha \PsiastX$.
\end{myproofx}
\par
%
Now let us return to the ordering criterion in terms of the supermodular 
order $\ordersm$ stated in Remark~\ref{rem:12}. The invariance of $\ordersm$ 
under non-decreasing component transformations 
allows to transfer these criteria to copula models. 
Furthermore, since we are interested in the ordering of the asymptotic 
dependence structures represented by the canonical spectral measures, 
$\Psiastone$ and $\Psiasttwo$, we can take 
any copulas that yield $\Psiastone$ and $\Psiasttwo$ as asymptotic 
dependence structures. 
\par
A natural choice is given by the \emph{extreme value copulas}, defined 
as the copulas of \emph{simple max-stable distributions} corresponding to 
$\Psiast_i$, i.e., the distributions 
\begin{equation}\label{eq:apl.4}
\Gast_i(x):=\exp\robrfl{-\nuast_i\robrfl{-[\infty,x]\setcomp}}
,\quad
x\in\Rplusd
\end{equation}
where $\nuast_i$ is the canonical exponent associated with $\Psiast_i$ 
via~\eqref{eq:apl.3}. For further details on max-stable and simple max-stable 
distributions we refer to \citet{Resnick:1987}.
Since extreme value copulas and canonical spectral measures can be 
considered as  
alternative parametrizations of the same asymptotic dependence structures, 
we obtain the following result. 
%
\par
\begin{proposition}
\label{prop:4.4}
Let $\Psiastone$ and $\Psiasttwo$ be canonical spectral measures on $\Simpd$. 
Further, for $i=1,2$, let $C_i$ denote the copula of the 
simple max-stable distribution $\Gast_i$ induced by $\Psiast_i$ 
according to~\eqref{eq:apl.4} and~\eqref{eq:apl.3}. 
Then $C_1 \ordersm C_2$ implies
\begin{enumerate}[(a)]
\item
$\Psiastone \orderGcalalpha \Psiasttwo$ for $\alpha\in(1,\infty)$;
\vspace{0.5em} 
\item
$\Psiasttwo \orderGcalalpha \Psiastone$ for $\alpha\in(0,1)$.
\end{enumerate}
\end{proposition}
\par
\begin{myproof}
Let $\nuast_i$ denote the canonical exponent measures corresponding
to $\Psiast_i$ and $\Gast_i$.
It is easy to see that the transformed measures 
\[
\nualphai:=\nuast_i\circ T\inv
,\quad i=1,2,
\]
with $\alpha>0$ and the transformation $T$ defined as
\[
T : x \mapsto 
\robrfl{\robrfl{x\compi}\powonebyalpha,\ldots,\robrfl{x\compd}\powonebyalpha}
,
\quad
x\in\Rplusd, 
\]
exhibit the scaling property with index $-\alpha$:
\begin{align*}
\nualphai\robr{tA} 
&=
t\powminusalpha \nualphai(A)
,\quad
A\in\Borel\robr{\Rplusd\setminus\cubr{0}}
\ldotp
\end{align*}
Hence the transformed distributions 
\begin{equation}
\label{eq:320}
\Galphai(x)
:=
\Gast_i\circ T\inv(x) 
= 
\exp\robrfl{-\nualphai\robrfl{[0,x]^c}}
\end{equation}
are max-stable with exponent measures $\nualphai$.
\par
It is well known that max-stable distributions with identical heavy-tailed margins are multivariate regularly varying \citep[cf.][]{Resnick:1987}.
Moreover, the limit measure $\nu$ in the multivariate regular variation condition can be chosen equal to the exponential measure associated with the property of max-stability.  
Consequently, the probability distributions $\Galphai$ for $i=1,2$ and $\alpha>0$ are multivariate regularly varying with tail index $\alpha$ and canonical spectral measures $\Psiast_i$. 
\par
Furthermore, it is easy to see that $X\sim\Galphaone$ and $Y\sim\Galphatwo$ 
have identical margins: 
\[
X\compi \disteq Y\compj
,\quad
i,j\in\cubr{1,\ldots,d}
\ldotp
\]
Moreover, due to the invariance of $\ordersm$ under non-decreasing marginal 
transformations, $C_1 \ordersm C_2$ implies 
\[
\Galphaone \ordersm \Galphatwo
\]
for all $\alpha>0$. 
Thus an application of the ordering criteria from Remark~\ref{rem:12} 
to $X\sim\Galphaone$ and $Y\sim\Galphatwo$ completes the proof.
\end{myproof}
%
%
\section{Examples}
\label{sec:5}
This section concludes the paper by a series of examples with parametric models illustrating the results from the foregoing sections.
Examples~\ref{ex:6} and \ref{ex:2} demonstrate application of Proposition~\ref{prop:4.4} to copula based models and the phenomenon of phase change for random vectors in $\Rplusd$.
The fact that the phase change does not necessarily occur in the general case is demonstrated by multivariate Student-t distributions in Example~\ref{ex:3}.
\par
\begin{example}
\label{ex:6}
Recall the family of Gumbel copulas given by
\begin{equation}
C_\theta(u)
:=
\exp\robrfl{- \robrfl{\sumioned\robrfl{-\log u\compi}^\theta}^{1/\theta}}
,\quad
\theta\in[1,\infty)
\ldotp
\end{equation}
Gumbel copulas are extreme value copulas, i.e., they are copulas of simple max-stable distributions. 
According to \citet{Hu/Wei:2002}, Gumbel copulas with dependence parameter $\theta\in[1,\infty)$ are ordered by $\ordersm$:
\begin{equation}
\label{eq:322}
\forall \theta_1,\theta_2\in[1,\infty) 
\quad
\theta_1 \le \theta_2
\impl
C_{\theta_1} \ordersm C_{\theta_2}
\ldotp
\end{equation}
Consequently, Proposition \ref{prop:4.4} applies to the family of 
canonical spectral measures $\Psiast_\theta$ corresponding to 
the Gumbel copulas $C_\theta$.
Thus $1\le\theta_1\le\theta_2<\infty$ implies 
$\Psiast_{\theta_1} \orderGcalalpha \Psiast_{\theta_2}$ for $\alpha>1$ 
and there is a phase change when $\alpha$ crosses the value $1$, i.e., 
for $\alpha\in(0,1)$ there holds 
$\Psiast_{\theta_2} \orderGcalalpha \Psiast_{\theta_1}$.
\par
Applying Theorem~\ref{theo:3.4}, one obtains ordering with respect to $\orderapl$ for random vectors $X$ and $Y$ on $\Rplusd$ that are multivariate regularly varying with canonical spectral measures 
of Gumbel type and have balanced tails ordered by $\orderapl$. 
In particular, 
this is the case if $X$ and $Y$ have identical 
regularly varying marginal distributions and 
Archimedean copulas that satisfy appropriate 
regularity conditions 
\citep[cf.][]{Genest/Rivest:1989, Barbe/Fougeres/Genest:2006}.
\par 
Moreover, it is also worth a remark that  
multivariate regularly varying random vectors with Archimedean copulas 
can only induce extreme value copulas of Gumbel type 
\citep[cf.][]{Genest/Rivest:1989}.  
\par
Figure~\ref{figure:32} illustrates the resulting diversification effects 
in the bivariate case, 
including indifference to portfolio diversification for $\alpha=1$ and 
the phase change occurring when $\alpha$ crosses this critical value. 
The graphics show the function 
$\xi\compone\mapsto\Psiast_\theta\,\gxialpha$ for selected values of 
$\theta$ and $\alpha$. 
Due to $X\in\Rplusd$, representation 
$\Psiast_\theta\,\gxialpha=\gammaxi/(\gammaeone + \gammaminuseone)$ simplifies to 
$\Psiast_\theta\,\gxialpha=\gammaxi/\gammaeone$
and therefore 
\[
\Psiast_\theta \, g_{e_1,\alpha} = \Psiast_\theta \, g_{e_2,\alpha} = 1
\ldotp
\]
\par
\begin{figure}
\centering
\subfigure[Varying $\alpha$ for $\theta=1.4$]
{\includegraphics[width=\mysubgraphicwidth]{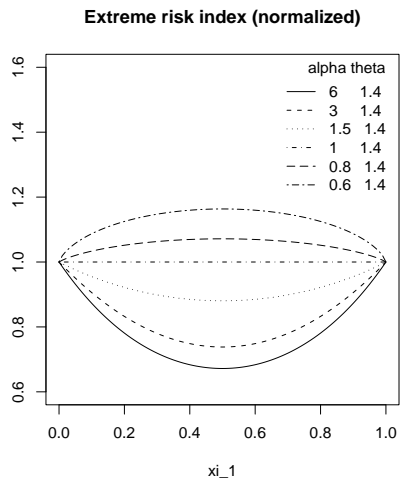}}
\subfigure[Varying $\alpha$ for $\theta=2$]
{\includegraphics[width=\mysubgraphicwidth]{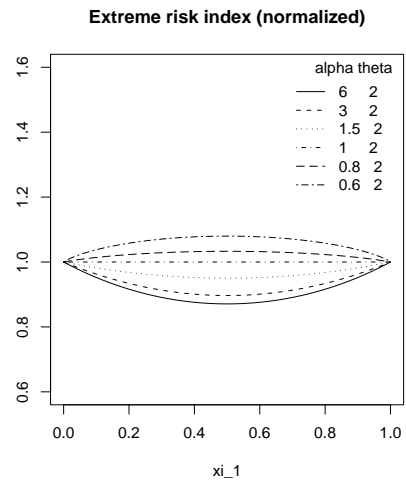}}
\\
\subfigure[Varying $\theta$ for $\alpha>1$]
{\includegraphics[width=\mysubgraphicwidth]{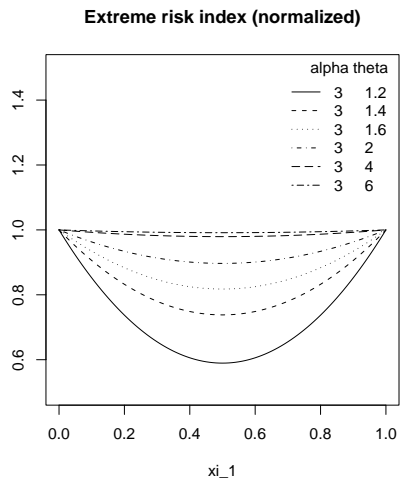}}
\subfigure[Varying $\theta$ for $\alpha<1$]
{\includegraphics[width=\mysubgraphicwidth]{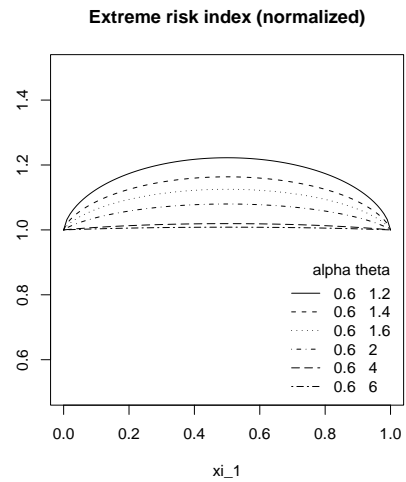}}
%
\caption{Bivariate Gumbel copulas: Diversification effects represented by functions $\xi\compone\mapsto\Psiast_\theta \,\gxialpha$ for selected values of $\theta$ and $\alpha$.}
\label{figure:32}
\end{figure}
\end{example}
As already mentioned above, Theorem~\ref{thm:5} generalizes some 
arguments from \citet{Embrechts/Neslehova/Wuethrich:2009}. 
The next example concerns Galambos copulas as addressed in 
that original publication. 
%
\par
\begin{example}
\label{ex:2}
Another family of extreme value copulas that are ordered by $\ordersm$ 
is the family of $d$-dimensional 
\emph{Galambos copulas} 
with parameter $\theta\in(0,\infty)$:
\begin{equation}
C_\theta(u)
:=
\exp\robrfl{\sum_{I\subset\cubr{1,\ldots,d}} 
(-1)^{\abs{I}}\robrfl{\sum_{i\in I} \robrfl{-\log u\compi}^{-\theta}}^{-1/\theta}
}
\ldotp
\end{equation}

According to \citet{Hu/Wei:2002}, 
$\theta_1 \le \theta_2$ implies $C_{\theta_1} \ordersm C_{\theta_2}$.
Thus Proposition~\ref{prop:4.4} yields ordering of the corresponding canonical spectral measures $\Psiast_\theta$ with respect to $\orderGcalalpha$. 
Similarly to the case of Gumbel copulas, $\theta_1\le\theta_2$ implies $\Psiast_{\theta_1} \orderGcalalpha \Psiast_{\theta_2}$ for $\alpha>1$  
and $\Psiast_{\theta_2} \orderGcalalpha \Psiast_{\theta_1}$ for $\alpha\in(0,1)$. 
%
%
\par
Finally, it should be noted that Galambos copulas correspond to the 
canonical exponent measures of random vectors $X$ in $\Rplusd$ with 
identically distributed regularly varying margins $X\compi$ and 
dependence structure of $-X$ given by an Archimedean copula with a 
regularly varying generator $\phi(1-1/t)$. Models of this type were 
discussed in recent studies of aggregation effects for extreme risks %
\citep[cf.][]{Alink/Loewe/Wuethrich:2004, %
Alink/Loewe/Wuethrich:2005, %
Embrechts/Neslehova/Chavez-Demoulin:2006, %
Barbe/Fougeres/Genest:2006, %
Embrechts/Lambrigger/Wuethrich:2008, %
Embrechts/Neslehova/Wuethrich:2009}.
\end{example}
\par
The final example illustrates results established in Proposition~\ref{prop:3.7} and Theorem~\ref{theo:2.4}. In particular, it shows that elliptical distributions do not exhibit a phase change at $\alpha=1$. 
\par
\begin{example}
\label{ex:3}
Recall multivariate Student-t distributions and 
consider the case with equal degrees of freedom, i.e., 
\begin{equation}
X\disteq \mu_X + R  A_X  U,
\quad
Y\disteq \mu_Y + R A_Y U,
\end{equation}
where $R\disteq\abs{Z}$ for a Student-t distributed random variable $Z$ with degrees of freedom equal to $\alpha\in(0,\infty)$. 
Further, let the generalized covariance matrices $C_X=C(\rho_X)$ and $C_Y=C(\rho_Y)$ be defined as 
\begin{equation}
\label{eq:146}
C(\rho):=
\left(
\begin{array}{cc}
1 &\rho\\
\rho & 1
\end{array}
\right)
\end{equation} 
and assume that $\rho_X \le \rho_Y$.
\par
As already mentioned in Remark~\ref{rem:2.6}(\ref{item:rem:2.6.a}), 
$C_X$ and $C_Y$ satisfy condition~\eqref{eq:3.26} and 
Proposition~\ref{prop:3.7} yields $X \orderapl Y$. Moreover, 
Proposition~\ref{prop:3.7} implies a uniform ordering of diversification 
effects in the sense that 
\[
\PsiastX=\Psiast_{\alpha,\rho_X} \orderGcalalpha \Psiast_{\alpha,\rho_Y}=\PsiastY
\]
for all $\alpha\in(0,\infty)$.
\par
Figure~\ref{figure:7} shows functions $\xi\compone \mapsto \Psiast_{\alpha,\rho}\,\gxialpha$ for selected parameter values $\rho$ and $\alpha$ that illustrate the ordering of asymptotic portfolio losses by $\rho$ and the missing phase change at $\alpha=1$. The indifference to portfolio diversification for $\alpha=1$ is also absent. 
Moreover, symmetry of elliptical distributions implies $\gammaminuseone = \gammaeone$ and, as a result,
\[
\Psiast_{\alpha,\rho}\, g_{e_1,\alpha} 
= 
\Psiast_{\alpha,\rho}\, g_{e_2,\alpha}  
=
1/2
\ldotp
\] 
Thus the standardization of the plots in Figure~\ref{figure:7} is different from that in Figure~\ref{figure:32}. 
\par
\begin{figure}
\centering
\subfigure[Varying $\alpha$ for $\rho>0$]
{\includegraphics[width=\mysubgraphicwidth]{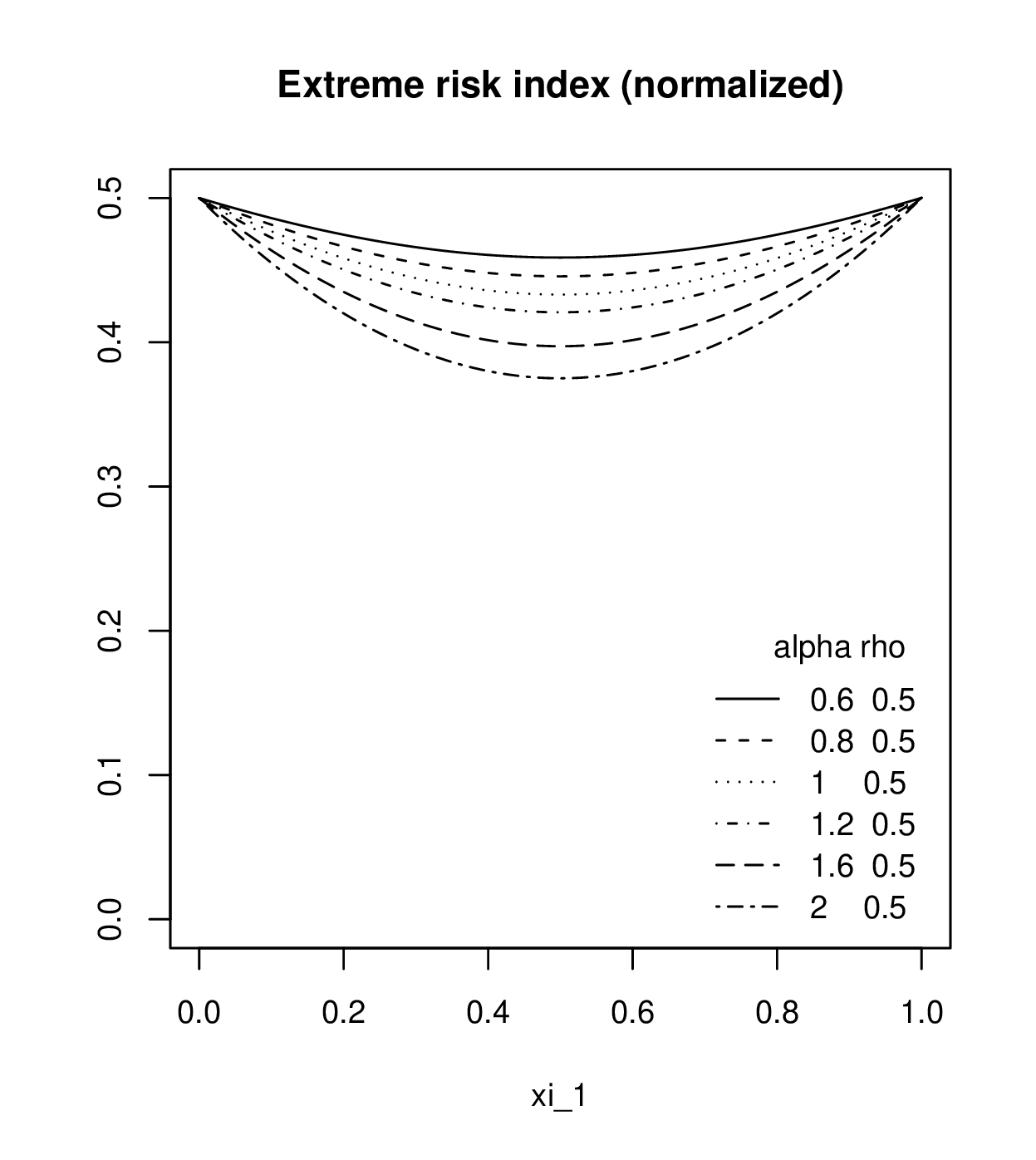}}
\subfigure[Varying $\alpha$ for $\rho<0$]
{\includegraphics[width=\mysubgraphicwidth]{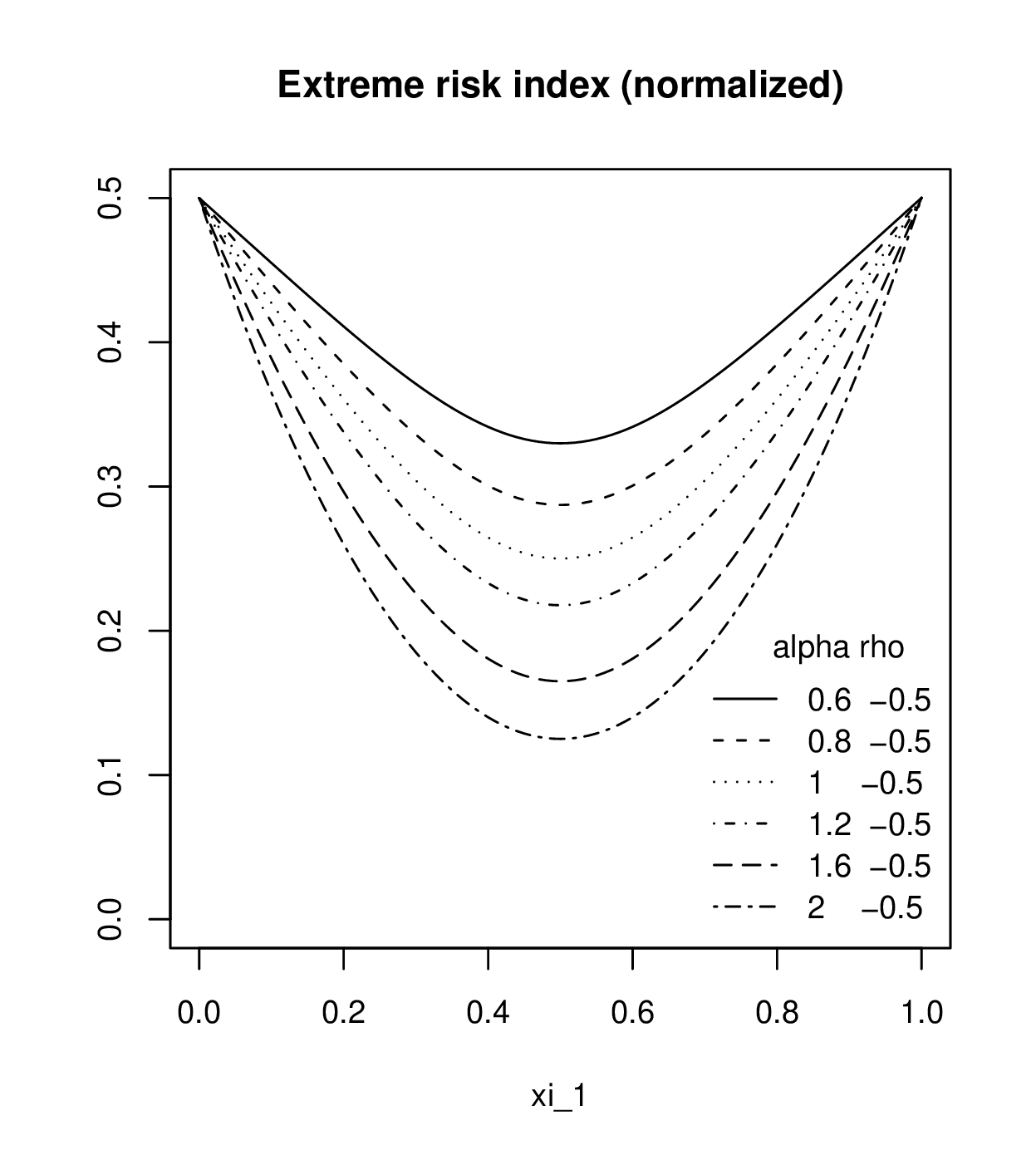}}
\\
\subfigure[Varying $\rho$ for $\alpha>1$]
{\includegraphics[width=\mysubgraphicwidth]{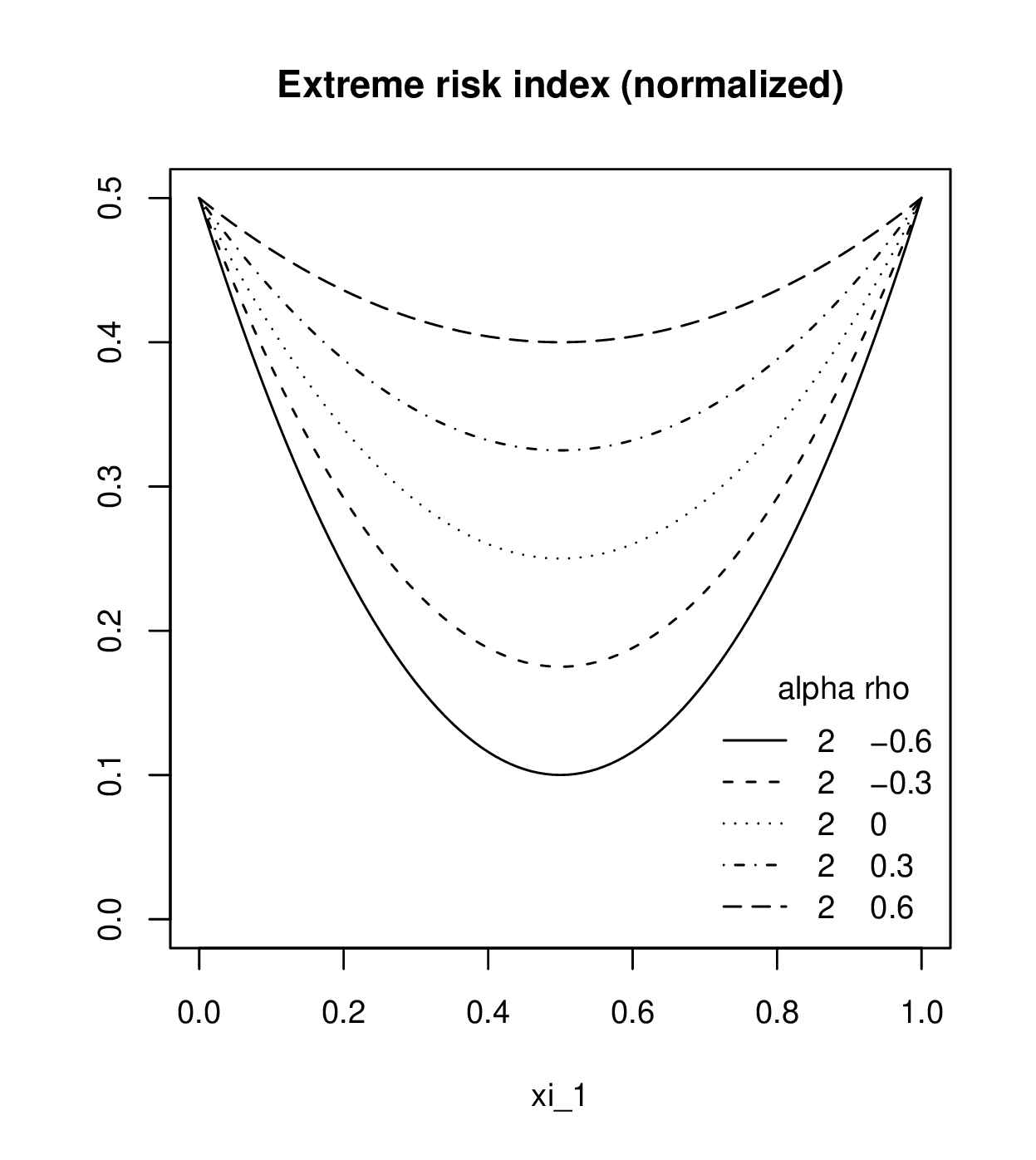}}
\subfigure[Varying $\rho$ for $\alpha<1$]
{\includegraphics[width=\mysubgraphicwidth]{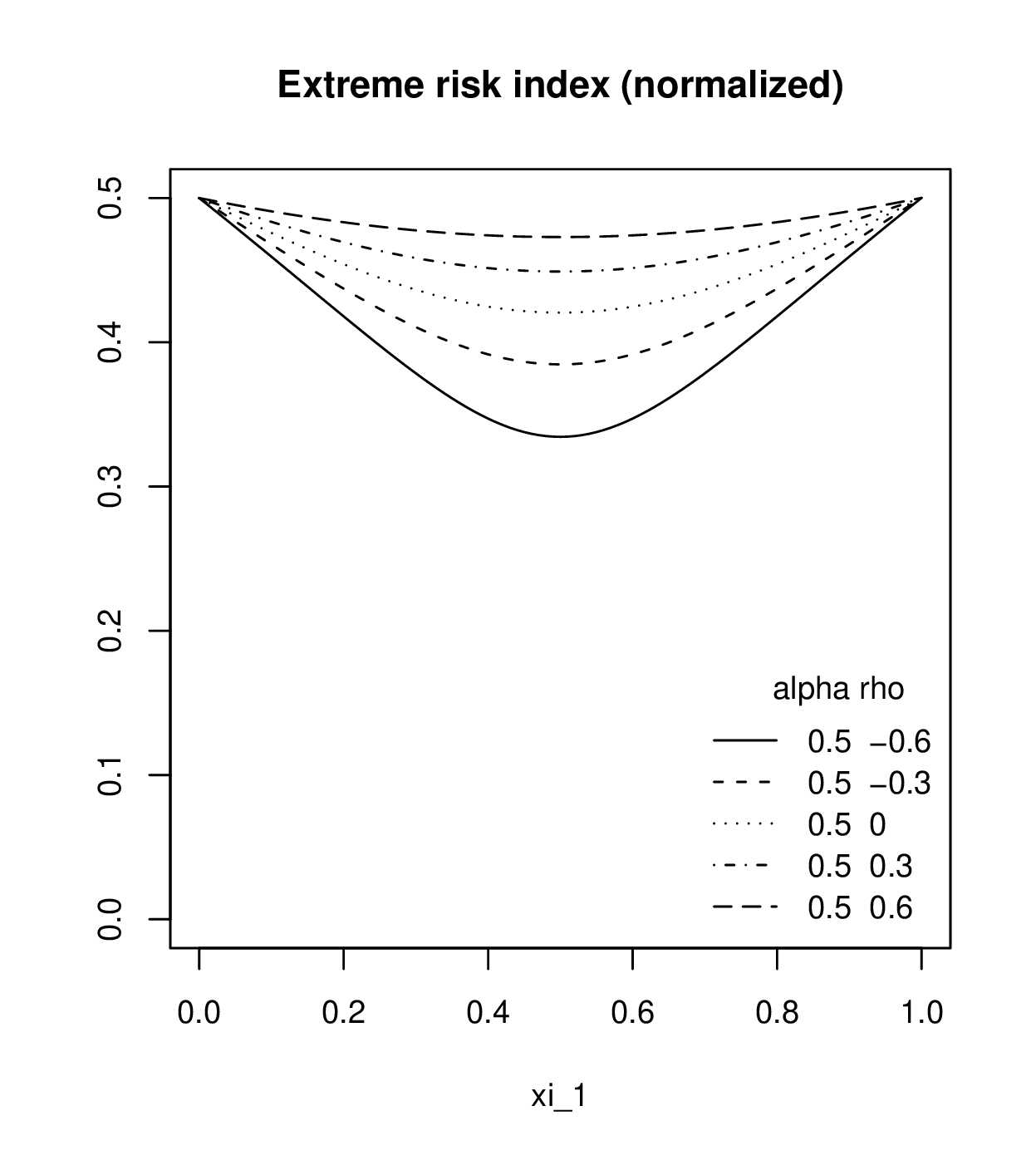}}
%
\caption{Bivariate elliptical distributions with generalized covariance matrices defined in~\eqref{eq:146}: Diversification effects represented by functions $\xi\compone\mapsto\Psiast_{\alpha,\rho}\, \gxialpha$ for selected values of $\rho$ and $\alpha$.}
\label{figure:7}
\end{figure}
\end{example}
\par
\begin{remark}
\label{rem:15}%
All examples the authors are aware of suggest that the diversification 
coefficient $\Psiast\gxialpha$ is decreasing in $\alpha$.
This means that risk diversification is stronger for lighter component tails 
than for heavier ones.
\par
However, it should be noted that the influence of the tail index $\alpha$ 
on risk aggregation is different from that. The asymptotic risk aggregation coefficient
\[
q_d := \lim_{t\toinf}\frac{\Prob\cubrfl{X\compone+\ldots+X\compd > t}}{\Prob\cubrfl{X\compone>t}}
\]
introduced by %
\citet{Wuethrich:2003} is known to be increasing in $\alpha$ when the loss components 
$X\compi$ are non-negative %
\citep[cf.][]{Barbe/Fougeres/Genest:2006}. 
It is easy to see that the restriction to non-negative $X\compi$ implies 
\[
q_d
=
\lim_{t\toinf}\frac{\Prob\cubrfl{\norm{X}_1 > t}}{\Prob\cubrfl{X\compone>t}}
=
\oneby{\gammaeone}
\ldotp
\]
Moreover, denoting the uniformly diversified portfolio by $\eta$, 
\[
\eta:=d\inv\robr{1,\ldots,1}
,
\]
one obtains 
\[
q_d 
= 
\lim_{t\toinf}\frac{\Prob\cubrfl{\eta\tr X > d\inv t}}{\Prob\cubrfl{X\compone>t}}
=
d\powalpha \frac{\gamma_\eta}{\gammaeone}
\ldotp
\]
Thus $q_d$ is a product of the factor $d\powalpha$, which is 
increasing in $\alpha$, and the ratio $\gamma_\eta/\gammaeone$, 
which is closely related to to the diversification coefficient 
$\Psiast\gxialpha$. 
\par
In particular, given equal marginal weights, i.e., 
\[
\gammaeone=\ldots=\gammaed,
\]
Proposition~\ref{prop:3.3}(\ref{item:L39.3}) yields 
\[
\frac{\gamma_\eta}{\gammaeone} = \Psiast g_{\eta,\alpha}
\ldotp
\]
As already mentioned above, the coefficients $\Psiast\gxialpha$ 
with $\xi\in\Simpd$ are decreasing in all examples considered here.
This means that the aggregation and the diversification of risks are 
influenced by the tail index $\alpha$ in different, maybe even always 
contrary ways.
\par
The question for the generality of this contrary influence is currently open. 
One can easily prove that the extreme risk index  $\gammaxi= \Psi\fxialpha$ is decreasing in $\alpha$ for $\xi\in\Simpd$. However, this result cannot be extended to $\Psiast\gxialpha$ directly  since $\Psiast\gxialpha$ is related to $\Psi\fxialpha$ by the normalizations~\eqref{eq:3.17} and~\eqref{eq:3.18}. 
%
The question 
whether $\Psiast\gxialpha$ with arbitrary $\xi\in\Simpd$  or at least 
$\Psiast g_{\eta,\alpha}$ is generally decreasing in $\alpha$ 
is an interesting subject for further research. 
\end{remark}
\section{Acknowledgements}
%
The research underlying this paper was done at the University of Freiburg. 
Georg Mainik would also like to thank RiskLab for financial support. 
%
%
\bibliography{apl-order}

\begin{thebibliography}{24}
\providecommand{\natexlab}[1]{#1}
\providecommand{\url}[1]{\texttt{#1}}
\expandafter\ifx\csname urlstyle\endcsname\relax
  \providecommand{\doi}[1]{doi: #1}\else
  \providecommand{\doi}{doi: \begingroup \urlstyle{rm}\Url}\fi

\bibitem[Alink et~al.(2004)Alink, L{\"o}we, and
  W{\"u}thrich]{Alink/Loewe/Wuethrich:2004}
Stan Alink, Matthias L{\"o}we, and Mario~V. W{\"u}thrich.
\newblock Diversification of aggregate dependent risks.
\newblock \emph{Insurance Math. Econom.}, 35\penalty0 (1):\penalty0 77--95,
  2004.
\newblock ISSN 0167-6687.
\newblock \doi{10.1016/j.insmatheco.2004.05.001}.

\bibitem[Alink et~al.(2005)Alink, L{\"o}we, and
  W{\"u}thrich]{Alink/Loewe/Wuethrich:2005}
Stan Alink, Matthias L{\"o}we, and Mario~V. W{\"u}thrich.
\newblock Analysis of the expected shortfall of aggregate dependent risks.
\newblock \emph{Astin Bull.}, 35\penalty0 (1):\penalty0 25--43, 2005.
\newblock ISSN 0515-0361.
\newblock \doi{10.2143/AST.35.1.583164}.

\bibitem[Anderson(1955)]{Anderson:1955}
T.~W. Anderson.
\newblock The integral of a symmetric unimodal function over a symmetric convex
  set and some probability inequalities.
\newblock \emph{Proc. Amer. Math. Soc.}, 6:\penalty0 170--176, 1955.
\newblock ISSN 0002-9939.

\bibitem[Barbe et~al.(2006)Barbe, Foug{\`e}res, and
  Genest]{Barbe/Fougeres/Genest:2006}
Philippe Barbe, Anne-Laure Foug{\`e}res, and Christian Genest.
\newblock On the tail behavior of sums of dependent risks.
\newblock \emph{Astin Bull.}, 36\penalty0 (2):\penalty0 361--373, 2006.
\newblock ISSN 0515-0361.
\newblock \doi{10.2143/AST.36.2.2017926}.

\bibitem[Basrak et~al.(2002)Basrak, Davis, and
  Mikosch]{Basrak/Mikosch/Davis:2002}
Bojan Basrak, Richard~A. Davis, and Thomas Mikosch.
\newblock A characterization of multivariate regular variation.
\newblock \emph{Ann. Appl. Probab.}, 12\penalty0 (3):\penalty0 908--920, 2002.
\newblock ISSN 1050-5164.

\bibitem[Boman and Lindskog(2009)]{Boman/Lindskog:2009}
Jan Boman and Filip Lindskog.
\newblock Support theorems for the {R}adon transform and {C}ram\'er-{W}old
  theorems.
\newblock \emph{J. Theor. Probab.}, 22\penalty0 (3):\penalty0 683--710, 2009.
\newblock ISSN 0894-9840.
\newblock \doi{10.1007/s10959-008-0151-0}.

\bibitem[de~Haan and Ferreira(2006)]{de_Haan/Ferreira:2006}
Laurens de~Haan and Ana Ferreira.
\newblock \emph{{E}xtreme {V}alue {T}heory}.
\newblock Springer Series in Operations Research and Financial Engineering.
  Springer, New York, 2006.
\newblock ISBN 978-0-387-23946-0; 0-387-23946-4.

\bibitem[Embrechts et~al.(2009{\natexlab{a}})Embrechts, Lambrigger, and
  W{\"u}thrich]{Embrechts/Lambrigger/Wuethrich:2008}
Paul Embrechts, Dominik~D. Lambrigger, and Mario~V. W{\"u}thrich.
\newblock Multivariate extremes and the aggregation of dependent risks:
  examples and counter-examples.
\newblock \emph{Extremes}, 12\penalty0 (2):\penalty0 107--127,
  2009{\natexlab{a}}.
\newblock ISSN 1386-1999.
\newblock \doi{10.1007/s10687-008-0071-5}.

\bibitem[Embrechts et~al.(2009{\natexlab{b}})Embrechts, Ne{\v{s}}lehov{\'a},
  and W{\"u}thrich]{Embrechts/Neslehova/Wuethrich:2009}
Paul Embrechts, Johanna Ne{\v{s}}lehov{\'a}, and Mario~V. W{\"u}thrich.
\newblock Additivity properties for value-at-risk under {A}rchimedean
  dependence and heavy-tailedness.
\newblock \emph{Insurance Math. Econom.}, 44\penalty0 (2):\penalty0 164--169,
  2009{\natexlab{b}}.
\newblock ISSN 0167-6687.
\newblock \doi{10.1016/j.insmatheco.2008.08.001}.

\bibitem[Fefferman et~al.(1972)Fefferman, Jodeit~jun., and
  Perlman]{Fefferman/Jodeit/Perlman:1972}
Charles Fefferman, Max Jodeit~jun., and Michael~D. Perlman.
\newblock {A spherical surface measure inequality for convex sets.}
\newblock \emph{Proc. Am. Math. Soc.}, 33:\penalty0 114--119, 1972.
\newblock \doi{10.2307/2038182}.

\bibitem[Genest and Rivest(1989)]{Genest/Rivest:1989}
Christian Genest and Louis-Paul Rivest.
\newblock A characterization of {G}umbel's family of extreme value
  distributions.
\newblock \emph{Stat. Probab. Lett.}, 8\penalty0 (3):\penalty0 207--211, 1989.
\newblock ISSN 0167-7152.

\bibitem[Hult and Lindskog(2002)]{Hult/Lindskog:2002}
Henrik Hult and Filip Lindskog.
\newblock Multivariate extremes, aggregation and dependence in elliptical
  distributions.
\newblock \emph{Adv. Appl. Probab.}, 34\penalty0 (3):\penalty0 587--608, 2002.
\newblock ISSN 0001-8678.
\newblock \doi{10.1239/aap/1033662167}.

\bibitem[Joe(1997)]{Joe:1997}
Harry Joe.
\newblock \emph{Multivariate Models and Dependence Concepts}, volume~73 of
  \emph{Monographs on Statistics and Applied Probability}.
\newblock Chapman \& Hall, London, 1997.
\newblock ISBN 0-412-07331-5.

\bibitem[Mainik(2010)]{Mainik:2010}
Georg Mainik.
\newblock \emph{On Asymptotic Diversification Effects for Heavy-Tailed Risks}.
\newblock PhD thesis, University of Freiburg, 2010.
\newblock URL
  \url{http://www.freidok.uni-freiburg.de/volltexte/7510/pdf/thesis.pdf}.

\bibitem[Mainik and R{\"u}schendorf(2010)]{Mainik/Rueschendorf:2010}
Georg Mainik and Ludger R{\"u}schendorf.
\newblock On optimal portfolio diversification with respect to extreme risks.
\newblock \emph{Finance Stoch.}, accepted for publication, 2010.
\newblock \doi{10.1007/s00780-010-0122-z}.

\bibitem[M{\"u}ller and Stoyan(2002)]{Mueller/Stoyan:2002}
Alfred M{\"u}ller and Dietrich Stoyan.
\newblock \emph{Comparison Methods for Stochastic Models and Risks}.
\newblock Wiley Series in Probability and Statistics. John Wiley \& Sons Ltd.,
  Chichester, 2002.
\newblock ISBN 0-471-49446-1.

\bibitem[Ne{\v{s}}lehov{\'{a}} et~al.(2006)Ne{\v{s}}lehov{\'{a}}, Embrechts,
  and Chavez-Demoulin]{Embrechts/Neslehova/Chavez-Demoulin:2006}
Johanna Ne{\v{s}}lehov{\'{a}}, Paul Embrechts, and Valerie Chavez-Demoulin.
\newblock Infinite-mean models and the {LDA} for operational risk.
\newblock \emph{J. Operational Risk}, 1\penalty0 (1):\penalty0 3--25, 2006.

\bibitem[Resnick(1987)]{Resnick:1987}
Sidney~I. Resnick.
\newblock \emph{{E}xtreme {V}alues, {R}egular {V}ariation, and {P}oint
  {P}rocesses}, volume~4 of \emph{Applied Probability. A Series of the Applied
  Probability Trust}.
\newblock Springer-Verlag, New York, 1987.
\newblock ISBN 0-387-96481-9.

\bibitem[Resnick(2007)]{Resnick:2007}
Sidney~I. Resnick.
\newblock \emph{{H}eavy-{T}ail {P}henomena}.
\newblock Springer Series in Operations Research and Financial Engineering.
  Springer, New York, 2007.
\newblock ISBN 978-0-387-24272-9; 0-387-24272-4.

\bibitem[Shaked and Shanthikumar(1994)]{Shaked/Shanthikumar:1994}
Moshe Shaked and J.~George Shanthikumar.
\newblock \emph{Stochastic Orders and Their Applications}.
\newblock Probability and Mathematical Statistics. Academic Press Inc., Boston,
  MA, 1994.
\newblock ISBN 0-12-638160-7.

\bibitem[Shaked and Shanthikumar(1997)]{Shaked/Shanthikumar:1997}
Moshe Shaked and J.~George Shanthikumar.
\newblock Supermodular stochastic orders and positive dependence of random
  vectors.
\newblock \emph{J. Multivariate Anal.}, 61\penalty0 (1):\penalty0 86--101,
  1997.
\newblock ISSN 0047-259X.
\newblock \doi{10.1006/jmva.1997.1656}.

\bibitem[Tong(1980)]{Tong:1980}
Yung~Liang Tong.
\newblock \emph{Probability inequalities in multivariate distributions}.
\newblock Academic Press [Harcourt Brace Jovanovich Publishers], New York,
  1980.
\newblock ISBN 0-12-694950-6.
\newblock Probabilities and Mathematical Statistics.

\bibitem[Wei and Hu(2002)]{Hu/Wei:2002}
Gang Wei and Taizhong Hu.
\newblock Supermodular dependence ordering on a class of multivariate copulas.
\newblock \emph{Stat. Probab. Lett.}, 57\penalty0 (4):\penalty0 375--385, 2002.
\newblock ISSN 0167-7152.
\newblock \doi{10.1016/S0167-7152(02)00094-9}.

\bibitem[W{\"u}thrich(2003)]{Wuethrich:2003}
Mario~V. W{\"u}thrich.
\newblock Asymptotic value-at-risk estimates for sums of dependent random
  variables.
\newblock \emph{Astin Bull.}, 33\penalty0 (1):\penalty0 75--92, 2003.
\newblock ISSN 0515-0361.
\newblock \doi{10.2143/AST.33.1.1040}.

\end{thebibliography}
\end{document}